\documentclass{article}

\usepackage[draft]{myStyle}
\usepackage{generic}
\usepackage{constraints}
\usepackage{bbm}

\newtheorem{theorem}{Theorem}[section]
\newtheorem{lemma}{Lemma}[section]
\newtheorem{proposition}{Proposition}[section]
\newtheorem{corollary}{Corollary}[section]

\newcommand{\mytitle}{Hyperboloidal initial data without logarithmic singularities}

\RequirePackage{hyperref}
\hypersetup
{
  colorlinks=true,
  linktocpage=true,
  linkcolor={red},
  anchorcolor={black},
  citecolor={green},
  menucolor={red},
  unicode=true,
  pdftitle={\mytitle},
  pdfauthor={K. Z. Csukás, I. Rácz}
}

\usepackage[backend=biber,style=numeric-comp,sorting=none,backref=true]{biblatex}
\title{\mytitle}
\author{{Károly Csukás\,\orcidlink{0000-0002-2408-1103}\footnote{E-mail address: {\tt csukas.karoly@wigner.hu}}}\ \ \ and \ 
{István Rácz\,\orcidlink{0000-0002-1991-0472}\footnote{E-mail address: {\tt racz.istvan@wigner.hu}}}}

\affil{HUN-REN Wigner RCP, H-1121 Budapest, Konkoly Thege Mikl\'{o}s \'{u}t  29-33, Hungary}

\begin{document}
  \maketitle

  \begin{abstract}
  	Andersson and Chru\'sciel showed that generic asymptotically hyperboloidal initial data sets admit polyhomogeneous expansions, and that only a non-generic subclass of solutions of the conformal constraint equations is free of logarithmic singularities. The purpose of this work is twofold. First, within the evolutionary framework of the constraint equations, we show that the existence of a well-defined Bondi mass brings the asymptotically hyperboloidal initial data sets into a subclass whose Cauchy development guaranteed to admit a smooth boundary, by virtue of the results of Andersson and Chru\'sciel. Second, by generalizing a recent result of Beyer and Ritchie, we show that the existence of well-defined Bondi mass and angular momentum, together with some mild restrictions on the free data, implies that the generic solutions of the parabolic-hyperbolic form of the constraint equations are completely free of logarithmic singularities. We also provide numerical evidence to show that in the vicinity of Kerr, asymptotically hyperboloidal initial data without logarithmic singularities can indeed be constructed.
  \end{abstract}

  \tableofcontents

	\parskip 5pt

  \section{Introduction}
  \label{sec:intro}
  
  Like many other useful concepts in physics, an isolated self-gravitating system is also an abstraction. An isolated self-gravitating system may consist of stars and black holes, and it is assumed to be so far away from other such systems that we can essentially ignore their influence, except perhaps for their gravitational radiation effects.
  There was a long evolution of the underlying ideas until Penrose finally introduced the precise mathematical model of asymptotically simple spacetimes, given as follows \cite{Penrose:1962ij, Frauendiener:2000mk, Kroon:2016ink}.

  Consider a  smooth spacetime $(M,g)$ representing a self-gravitating isolated system. Such a spacetime is called {\it asymptotically simple} if there exists a smooth spacetime $(\widetilde M,\widetilde g)$ with boundary ${\mathscr{I}}\not=\emptyset$ such that $M$ can be diffeomorphically identified with the interior, $\widetilde M\setminus \mathscr{I}$, of $\widetilde M$  so that
  	\vskip-.2cm
  	\begin{equation}
  		\widetilde g= \Omega^2 g \ {\rm on} \ M\,,
  	\end{equation}
  	where $\Omega$ is a smooth boundary defining function on $\widetilde M$, i.e.
  	\vskip-.2cm
  	\begin{equation}
  	\Omega>0 \  {\rm on} \ M\,, \ \  {\rm and\ also}  \quad  \Omega=0 \ \ {\rm and} \ \ d\Omega\not=0 \  {\rm on}\  \mathscr{I}\,.
  	\end{equation}
  	
  	Of course, this characterization of self-gravitating isolated systems does not refer to special coordinate systems, and the real strength lies in making the key geometric structures transparent. It also follows from this definition that $g$-null geodesics are complete in those directions in which they approach $\mathscr{I}$.
  	The degree of smoothness of the rescaled metric $\widetilde g$ is critical, since the falloff behavior of the physical fields depends on this smoothness \cite{Penrose:1962ij,Friedrich:1991nn, Friedrich:2015bax}.
  	
  	Having made the above general determination, natural questions come to mind. Are the asymptotic conditions compatible with the behavior of the gravitational field in a ``sufficiently large'' class of physically realistic situations? Is it reasonable to ask whether there are solutions to Einstein's vacuum equations other than Minkowski spacetime that satisfy these conditions?
  	
  	In the course of answering these questions, the hyperboloidal initial value problem turned out to be the most appropriate tool for doing so \cite{Frauendiener:2000mk, Kroon:2016ink}. Friedrich developed a powerful formalism for studying asymptotically simple spacetimes. His conformal field equations have been used to study the evolution of suitably regular hyperboloidal data. Friedrich also proved that smooth data evolve into solutions which satisfy the conditions in the definition of asymptotically simple spacetimes. Moreover, the developments admit a conformally regular point $i^+$ at time-like infinity, analogous to the point $i^+$ in the case of Minkowski space, if the data are sufficiently close to Minkowskian hyperboloidal data  \cite{Friedrich:1991nn, Friedrich:2015bax}.

  	In the hyperboloidal initial value problem, data are prescribed on a spacelike hypersurface in an asymptotically simple spacetime that extends to future null infinity.
  	
  	If one wants to study solutions of the conformal field equations, the first task is to construct hyperboloidal initial data. However, constructing asymptotically hyperboloidal initial data in such a way that they are regular at future null infinity, $\scri$, is a non-trivial task. To see this,  it is worth recalling that by virtue of the provision used in \cite{Andersson:1992yk,Andersson:1993we,Andersson:1994ng,Andersson:1996xd}, $(\Sigma, h_{ab}, K_{ab})$ is an {\it asymptotically hyperboloidal data set} (not necessarily a solution of the vacuum constraints) if there exists a triple $(\widetilde{\Sigma},\widetilde{\omega},\Phi)$ such that:
  	\begin{itemize}
  		\item[(i)] $\widetilde{\Sigma}$ is a manifold with boundary $\partial\widetilde{\Sigma}$, and there exists an embedding $\Phi$ of $\Sigma$ into $\widetilde{\Sigma}$ such that $\Phi$ is a diffeomorphism between $\Sigma$ and $\Phi(\Sigma)=\widetilde{\Sigma}\setminus \partial\widetilde{\Sigma}$. 
  		\item[(ii)] $\widetilde{\omega}:\widetilde{\Sigma}\rightarrow \mathbb{R} $ is a smooth boundary defining function such that $\widetilde{\omega}>0$ on $\widetilde{\Sigma}\setminus\partial\widetilde{\Sigma}$ and $\widetilde{\omega}=0$, $d\widetilde{\omega}\not=0$ on  $\partial\widetilde{\Sigma}$.
  		\item[(iii)] The embedding $\Phi$ is such that $\widetilde{\omega}^2 \Phi{}{}^* h$ is a Riemannian metric on $\widetilde{\Sigma}\setminus\partial\widetilde{\Sigma}$ which extends as a regular Riemannian metric to  $\partial\widetilde{\Sigma}$\,.
  		\item[(iv)] The push forward of the trace $K=h^{ab}K_{ab}$ is bounded away from zero near $\partial\widetilde{\Sigma}$.
  		\item[(v)] Let $L_{ab}$ be the trace-free part of $K_{ab}$
  		. Then the field $\widetilde{\omega}^{} \Phi^{}{}^* L_{ab}$,
  		defined on $\widetilde{\Sigma}\setminus\partial\widetilde{\Sigma}$, extends as a regular tensor field to $\partial\widetilde{\Sigma}$\,.
  	\end{itemize}

  	We note here that the specific smoothness requirements depend on the application. We also want to emphasize that the above definition says nothing about solving the constraints. To obtain an asymptotically hyperboloidal {\it initial data} set, we choose some free data that satisfies the above requirements and then solve the constraints for the relevant constrained variables.
  	  	
  	Using the conformal method, Andersson and Chru\'sciel showed \cite{Andersson:1992yk,Andersson:1993we,Andersson:1994ng, Andersson:1996xd} that generic solutions of the elliptic form of the constraints (applying a constant mean curvature slicing) admit polyhomogeneous expansions at $\partial\widetilde{\Sigma}$, i.e. the data have asymptotic expansions in terms of powers of $\widetilde{\omega}$ and $\log \widetilde{\omega}$. They also showed that to avoid the presence of logarithmic terms in the asymptotic expansions, it is not enough to require that the free data be smoothly expandable over $\partial\widetilde{\Sigma}$. In order to get rid of the logarithmic singularities, it is also necessary to impose rather strong conditions on the conformally rescaled extrinsic curvature. The main implication of these observations is that generic initial data constructed by the conformal method are not regular enough to avoid the involvement of logarithmic terms in the asymptotic expansions.
  	
  	It is then natural to ask whether the assumption that the conformal field equations hold up to and including $\mathscr{I}$ is consistent with the smoothness assumption made in the definition of asymptotically simple spacetimes. As pointed out by Friedrich in \cite{Friedrich:2015bax} (see also the discussions in \cite{Andersson:1993we}): The question is not whether $C^\infty$ should be replaced by $C^k$ for some large $k$. The question is whether solutions of the field equations admit conformal extensions of class $C^k$, where $k$ can be chosen large enough to make the concept of asymptotically simple spacetimes meaningful, and to guarantee that the conformal Weyl tensor tends to zero at the conformal boundary.
  	If one has initial data whose asymptotic expansion contains logarithmic terms, then the evolving metric should also be contaminated by  these type of logarithmic terms \cite{Andersson:1992yk,Andersson:1993we,Andersson:1994ng, Andersson:1996xd}. This would discredit the numerical simulations, which in general cannot handle the non-smoothness of the null boundary, and also because there would be no way to decompose the physical metric into a smooth ``non-physical metric'' and a conformal factor.

  	It is worth noting that there are numerical adaptations of the conditions presented in \cite{Andersson:1992yk,Andersson:1993we,Andersson:1994ng,Andersson:1996xd} which led to various types of initial data within the conformal framework that are regular at future null infinity. For instance, the authors in \cite{Buchman:2009ew,Bardeen:2011pd} presented conformally flat binary black hole initial data requiring constancy of mean curvature, while \cite{Schinkel:2013zm} constructed perturbed Kerr initial data using the more permissive asymptotically constant mean curvature gauge. Friedrich's conformal field equations were employed in numerical investigations of asymptotically flat spacetimes in the vicinity of null infinity \cite{Frauendiener:1997zc,Frauendiener:1997ze,Frauendiener:1998ud,Hubner:1999th,Hubner:1998hn,Frauendiener:1998yi,Frauendiener:2002ix}.
  	
  	At the stage of this difficulty, Beyer and Ritchie \cite{Beyer:2021kmi} came up with an interesting alternative view  and with a powerful argument. They claim that if on a hyperboloidal initial data surface there exist global smooth  asymptoticaly constant mean curvature solutions to the parabolic-hyperbolic form of the constraints (in physical spacetime), furthermore these solutions are guaranteed to extend to $\partial\widetilde{\Sigma}$ to some finite order, then they extend smoothly to $\partial\widetilde{\Sigma}$. In fact, in \cite{Beyer:2021kmi} it is shown that if the free data satisfies a suitable set of falloff conditions and a simple algebraic condition, then finitely regular smooth solutions of the parabolic-hyperbolic form of the constraints extend smoothly to null infinity, thereby producing hyperboloidal initial data sets free of logarithmic singularities.
 
    With this paper we overbridge the gap between the statements in \cite{Andersson:1992yk,Andersson:1993we,Andersson:1994ng, Andersson:1996xd} and those in \cite{Beyer:2021kmi}, originating from the different methods, formalisms, and focuses of the papers. First, we show within the evolutionary framework of constraint equations (see the first main result of the present paper, Theorem \ref{main theorem 1} and Corollary \ref{corollary}) that by requiring the existence of a well-defined Bondi mass\footnote{Technically, what we call "Bondi mass" is actually a component of "Bondi momentum" related to time translation. Therefore, it is more accurate to refer to it as "Bondi energy." The same argument applies to Hawking mass and ADM mass. According to the discussions in \cite{Szabados:2009eka}, we should refer to them as Hawking energy and ADM energy, respectively. For historical and readability reasons, however, we will use the term "mass" in each case.
    It is also worth noting that the finiteness of Hawking energy (see relations (6.2) and (6.3) in \cite{Szabados:2009eka}) implies the finiteness of Hawking four-momentum. This, in turn, implies the finiteness of Hawking mass. However, note that the finiteness of Hawking mass does not self-evidently imply the finiteness of energy or momentum. Investigating this issue would require careful consideration.} we obtain a subspace of the asymptotically hyperboloidal initial data sets  whose Cauchy development is guaranteed to admit smooth conformal completion, by virtue of the results of \cite{Andersson:1992yk,Andersson:1993we,Andersson:1994ng,Andersson:1996xd}. Second, we provide a substantial generalization of the result of Beyer and Ritchie \cite{Beyer:2021kmi} by proving that the existence of well-defined Bondi mass and angular momentum, together with some mild restrictions on the free data, implies that the generic solutions of the parabolic-hyperbolic form of the constraint equations are completely free of logarithmic singularities. In doing so we still economize the Fuchsian analysis  that was introduced in \cite{Beyer:2021kmi}. Finally, we provide numerical evidence that strongly excited near-Kerr asymptotically hyperboloidal initial data without logarithmic singularities can also be constructed using free data from a considerably larger class, and without enforcing the asymptotically constant mean curvature gauge condition.
  
    The outline of this paper is as follows. In Section \ref{sec:setup} we recall the parabolic-hyperbolic formulation of the vacuum constraint equations. In Section \ref{sec:hyper}, initial data with some minimal regularity and the main result of Andesson and Chru\'sciel are recalled. Section \ref{sec:main} contains our main analytic results. In Section \ref{sec:main-bondi} we discuss the implications of the assumption that well-defined Bondi mass and angular momentum can be associated with an asymptotically hyperboloidal initial data set that are subject the parabolic-hyperbolic form of the constraints. In Section \ref{sec:main-smooth} we also provide the aforementioned generalization of the result in \cite{Beyer:2021kmi}, and prove our second main result, Theorem \ref{main theorem 2}, which provides the smoothness of asymptotically hyperboloidal initial data. In particular, we outline the Fuchsian analyses used in Beyer and Ritchie's proof by pointing out the differences in the assumptions used in \cite{Beyer:2021kmi} and in this paper. Section \ref{sec:numerical} discusses the numerical construction of some highly excited near Kerr asymptotically hyperboloidal initial data. Section \ref{sec:background} explains the choice of free data, while Section \ref{sec:numsetup} introduces the numerical method used. Section \ref{sec:asymrel} explains the sophisticated methods that must be used to verify the smoothness of the asymptotically hyperboloidal initial data. In particular, Subsection \ref{subsec: approx Hawking} will highlight the difficulties in numerically evaluating the Hawking mass in the asymptotic region. Finally, Section \ref{sec:summ} summarizes our main results. There is also an appendix that collects some formulas that provide important clues to some of the arguments involved, but it would be inappropriate to include them in the main text.

  \section{The parabolic-hyperbolic form of the constraints}
  \label{sec:setup}

  In Einstein's theory the vacuum initial data consists of the triplet $(\Sigma,h_{ab},K_{ab})$, where $\Sigma$ is a $3$-dimensional differentiable manifold, $h_{ab}$ is a Riemannian metric on it, and $K_{ab}$ is a symmetric tensor field. The interpretation of the latter two is as follows: if $\Sigma$ is embedded in the $4$-dimensional spacetime obtained by evolution, $h_{ab}$ is the induced metric, while $K_{ab}$ is the extrinsic curvature of $\Sigma$. However this embedding is only possible if the triplet $(\Sigma,h_{ab},K_{ab})$ satisfies the vacuum constraint equations
  \begin{align}
    \Rthree-K_{ab}K^{ab}+K^2&=0,\label{eq:Ham}\\
    D_bK^b{}_a-D_aK&=0,\label{eq:Mom}
  \end{align}
  where $D_a$ is the covariant derivative compatible with $h_{ab}$, $\Rthree$ is the corresponding Ricci curvature scalar, $K=h^{ab}K_{ab}$ is the trace of the extrinsic curvature. 
  
  \subsection{The $2+1$ decomposition}\label{subsec: 2+1}
  
  The parabolic-hyperbolic form of the constraints \eqref{eq:Ham}-\eqref{eq:Mom} is derived using a $2+1$ decomposition of $\Sigma$ \cite{Racz:2015mfa,Racz:2014jra,Racz:2014gea}. Since our main concern will be the asymptotic behavior of asymptotically hyperboloidal initial data sets, we assume that $\Sigma$ is foliated by the $r=const$ level surfaces, $\mathscr{S}_{r}$, of a suitably smooth function $r: \Sigma\rightarrow \mathbb{R}^+$ which are diffeomorphic to $2$-spheres. Accordingly, we will assume that $\Sigma=(r_0,\infty)\times \mathbb{S}^2$ for some $r_0>0$. In the asymptotically hyperboloidal setup, we can choose the foliation-defining function $r$ (at least close to $\partial\widetilde{\Sigma}$) by the relation $r^{-1}=\widetilde{\omega}\circ\Phi$ on $\Sigma$. We will also parameterize the leaves of the foliation with $\omega=r^{-1}$, and mainly for historical reasons with $t (=\omega=r^{-1})$ in the proof of our second main result (see Theorem \ref{main theorem 2}). In a slight abuse of notation, we will occasionally replace $\widetilde{\omega}$ with $\omega$ or $r^{-1}$.

  In addition to the foliation selected above, we also choose a flow vector field $r^a$ on $\Sigma$ such that $r^aD_ar=1$. After constructing the unit norm $1$-form, $\nhat_a$, on $\Sigma$, normal to the level sets $\mathscr{S}_{r}$, we have the induced metric
  \begin{equation}
    \gammahat_{ab}=h_{ab}-\nhat_a\nhat_b\,,
  \end{equation}
  with $\gammahat_{a}{}^b=\gammahat_{ae} h^{eb}$ acting as a projector to $\mathscr{S}_{r}$.
  Now we can decompose $h_{ab}$ and $K_{ab}$ into their different projections. The decomposition of the evolution vector field gives the lapse, $\Nhat=\nhat_ar^a$, and the shift, $\Nhat_a=\gammahat_{ab}r^b$. The extrinsic curvature $K_{ab}$ also splits into $\kkappa=\nhat^a\nhat^bK_{ab}$, $\kk_a=\nhat^b\gammahat_a{}^cK_{bc}$, the $\gammahat$-trace $\KK=\gammahat^{ab}K_{ab}$, and the $\gammahat$-trace-free tensor $\Kc_{ab}=\gammahat_a{}^c\gammahat_b{}^dK_{cd}-\tfrac{1}{2}\gammahat_{ab}\KK$ on $\mathscr{S}_{r}$. In addition to these fundamental variables, we will also use the extrinsic curvature $\Kh_{ab}$ of the level surfaces $\mathscr{S}_r$ with respect to $\nhat^a$, together with its trace, $\Kh=\Kh_{ab}\gammahat^{ab}$, and the trace-free part $\circon{\widehat{K}}{}_{ab}=\Kh_{ab}-\tfrac{1}{2}\gammahat_{ab}\Kh$, and also $\Kstar_{ab}$ defined by the relation $\Kstar_{ab}=\Nhat\Kh_{ab}$.

  Finally we define the standard spherical coordinates $(\vartheta,\varphi)$ on one of the $r=const$ leaves of the foliation, say on $\mathscr{S}_{r_0}$, and construct a null dyad, $(q^a,\qbar^a)$, which has their indices raised and lowered by the unit sphere metric $q_{ab}=\,q_{(a}\overline{q}_{b)}$, and normalized as $q^a\overline{q}_a=2$ \cite{Racz:2017krc}. This dyad is then Lie propagated along the evolution vector field $r^a$. Our fundamental variables are the dyad components of the projected quantities introduced in the previous paragraph: the lapse $\Nhat$, the shift $\NN=q^a\Nhat_a$, the components of the induced $2$-metric, $\aaa=\tfrac{1}{2}q^a\qbar^b\gammahat_{ab}$ and $\bb=\tfrac{1}{2}q^aq^b\gammahat_{ab}$, and the projections of the extrinsic curvature, $\kkappa$, $\kk=q^a\kk_a$, $\KK$, and $\Kcqq=q^aq^b\Kc_{ab}$. Table \ref{tab:variables} lists the definitions of the most important spin-weighted variables.

  \begin{table}[h]
		\centering  \hskip-.15cm
		\begin{tabular}{|c|c|c|}
			\hline notation &  definition  & spin-weight \\ \hline \hline

			$\aaa$ &  $\tfrac12\,q^i\,\qbar^j\,\gammahat_{ij}$  &  $0$ \\  \hline

			$\bb$ &  $\tfrac12\,q^i q^j\,\gammahat_{ij}$
			&   $2$ \\  \hline

			$\dd$ &  $\aaa^2-\bb\,\bbbar$
			&   $0$ \\  \hline

			$\AAA$ &  $q^a q^b {C^e}{}_{ab}\,\qbar_e
			= \dd^{-1}\left\{ \aaa\left[2\,\eth\,\aaa
			-\,\ethb\,\bb\right]
			-  \,\bbbar\,\eth\,\bb \right\} $
			&   $1$ \\  \hline

			$\BB$ &  $\,\qbar^a q^b {C^e}{}_{ab}\,q_e
			= \dd^{-1}\left\{ \aaa\,\ethb\,\bb
			- \bb  \,\eth\,\bbbar\right\}$
			&   $1$ \\  \hline

			$\CC$ &  $q^a q^b {C^e}{}_{ab}\,q_e
			= \dd^{-1}\left\{ \aaa\,\eth\,\bb
			-  \bb\left[2\,\eth\,\aaa
			-\,\ethb\,\bb\right] \right\}$
			&   $3$ \\  \hline

			$\,\Rhat$ &  $\tfrac12\, {\aaa}^{-1}\left(2\, \RR
			- \left\{ \,  \eth\,\BBbar - \ethb\,\AAA
			- \tfrac12\,\left[\, \CC\,\CCbar
			- \BB\,\BBbar \,\right]\, \right\}\,\right)$
			&   $0$ \\  \hline

			${\NN}$ &  $q^i\Nhat_i$
			&   $1$ \\  \hline

			${\NNt}$ &  $q_i\Nhat^i=\dd^{-1}(\aaa\NN-\bb\NNbar)$
			&   $1$ \\  \hline

			$\kk$ &  $q^i {\kk}{}_{i}$  &
			  $1$
			\\  \hline

			$\KK$ &  $ \gammahat^{kl} \,{ K}{}_{kl}
			$  &   $0$ \\  \hline

			$\Kcqq$ &  $q^kq^l\,\Kc{}_{kl}
			$  &   $2$ \\  \hline

			$\Kcqqb$ &  $q^k\,\qbar^l\,\Kc{}_{kl}
			= (2\,\aaa)^{-1} [\,\bb\,\Kcqqbar
			+  \bbbar\,\Kcqq \,]
			$  &   $0$ \\  \hline

			$\,\Kstar$ &  $\gammahat^{ij} \Kstar_{ij} =\Nhat\,\gammahat^{ij}\Kh_{ij}=\Nhat\Kh $
			&   $0$ \\  \hline

			$\Kstarqq$ &  $q^i q^j\Kstar_{ij}
			= \tfrac12\,\left\{2\,\partial_r\bb - 2\,\eth\,\NN
			+ {\CC}\,\NNbar +\AAA \,{\NN} \,  \right\} $
			&   $2$ \\  \hline

			$\Kstarqqb$ &  $q^k\,\qbar^l\,\Kstar {}_{kl}
			=  {\aaa}^{-1}\{\,\dd\cdot\Kstar\}
			+ \tfrac12 \,[\,\bb\,\Kstarqqbar
			+  \bbbar\,\Kstarqq  \,]\,\}
			$  &   $0$ \\  \hline

		\end{tabular}
		\caption{\small The variables applied in providing the evolutionary form of the constraints.}\label{tab:variables}
	\end{table}
	
	Now consider an asymptotically hyperboloidal initial data set $(\Sigma, h_{ab},K_{ab})$ as specified in \cite{Andersson:1993we}.  
	Using the spin-weighted variables introduced above, a simple calculation shows that Proposition 1 of \cite{Beyer:2021kmi} can be reformulated as follows. 
	
	\begin{proposition}\label{prop1} 
		A vacuum initial data set $(\Sigma, h_{ab}, K_{ab})$ is asymptotically hyperboloidal (not necessarily a solution to the vacuum constraints) on $\Sigma=(r_0,\infty)\times \mathbb{S}^2$ if the following falloff conditions hold for the spin-weighted variables
		\begin{align}
			\hskip-0.0cm\Nhat&=\Nhat_{1}r^{-1} +\mathscr{O}(r^{-2})\,, \hskip0.4cm\NN=\mathscr{O}(r^{-1})\,,  \hskip0.2cm\aaa=r^{2}+ \mathscr{O}(r)\,,
			\hskip0.1cm\bb=\mathscr{O}(r)\,,  \label{eq:falloff1}\\
			\hskip-0.0cm\KK&=\KK_0+\mathscr{O}(r^{-1})\,, \hskip0.3cm\KK-2\kkappa=\mathscr{O}(r^{-1})\,,  \hskip0.2cm\kk=\mathscr{O}(1)\,,  \hskip0.4cm\hskip0.1cm\Kcqq=\mathscr{O}(r)\,,  \label{eq:falloff2}
		\end{align}
		where $\Nhat_{1}$ and $\KK_0$ are strictly positive smooth functions on $\partial\widetilde{\Sigma}$, which is the asymptotic limit of the foliating $\mathscr{S}_{r}$ level surfaces.
	\end{proposition}
	
	Note that, as in the definition of asymptotically hyperboloidal initial data in Section \ref{sec:intro}, no reference to the constraint equations was made here. Furthermore, the above falloff conditions only ensure the existence of well-defined $C^0$ limits of the involved fields at $\partial\widetilde{\Sigma}$. Note also that the falloff behavior of $\aaa$ and $\bb$ corresponds to the falloff behavior of the induced metric $\gammahat_{ab}$ on the $r=const$ level surfaces $\gammahat_{ab}= r^2\,q_{ab}+\mathscr{O}(r)$ used in Proposition 1 of \cite{Beyer:2021kmi}. To see that this is the appropriate falloff behavior for $\gammahat_{ab}$, first recall that by choosing an appropriate conformal gauge, the conformally rescaled two-metric on $\partial\widetilde{\Sigma}$, that is topologically a two-sphere, can be chosen without loss of generality to be the metric of the unit sphere $q_{ab}$ (for details see, e.g., chapter 11.1 of \cite{Wald:1984rg}). This then implies that in a sufficiently small neighborhood of $\partial\widetilde{\Sigma}$ on $\widetilde{\Sigma}$ the conformally rescaled two-metric, on the $r=const$ level surfaces, can be assumed to be of the form $q_{ab}+\omega \,p_{ab}(\omega)$, where $p_{ab}(\omega)$ is a smooth tensor field of the type $(0,2)$ there, so that $q_{ab}+\omega\,p_{ab}(\omega)$ is non-singular in the neighborhood under consideration. This, in turn, also implies that the two-metric $\gammahat_{ab}$, on the $r=const$ level surfaces, must have the asymptotic form $\gammahat_{ab}= r^2\,q_{ab}+\mathscr{O}(r)$. Exactly this falloff behavior of $\gammahat_{ab}$ is reflected by the falloff of $\aaa$ and $\bb$, the requirements we used above in \eqref{eq:falloff1}.
	 
	Note also that for an asymptotically hyperboloidal initial data set $(\Sigma, h_{ab},K_{ab})$, as specified in \cite{Andersson:1993we}, the trace of the three-dimensional extrinsic curvature, $K=h^{ab}K_{ab}=(\nhat^{a}\nhat^{b}+\gammahat^{ab})\,K_{ab}=\kkappa+\KK$, must tend to a strictly positive smooth function on $\partial\widetilde{\Sigma}\sim\mathbb{S}^2$. Since at leading order $\KK\sim 2\kkappa$ also holds we conclude that $\KK^{(0)}$ itself must be a strictly positive smooth function on $\partial\widetilde{\Sigma}$.
	Finally, note that, as it was also pointed out in \cite{Beyer:2021kmi}, and in accordance with the results of \cite{Andersson:1992yk,Andersson:1993we,Andersson:1994ng,Andersson:1996xd}, the falloff conditions \eqref{eq:falloff1}-\eqref{eq:falloff2} are sufficient to ensure that an initial data set $(\Sigma, h_{ab},K_{ab})$ is asymptotically hyperboloidal, but they are not necessary.
	
	\subsection{The parabolic-hyperbolic form of the constraints}
  
  The parabolic-hyperbolic interpretation, then, consists of solving equations \eqref{eq:Ham}-\eqref{eq:Mom} for $\Nhat$, $\kk$, and $\KK$ as constrained variables. Following this procedure the constraints read as \cite{Racz:2017krc}
  \begin{multline}
      \label{eq:phN}
      \Kstar\left[\partial_r\Nhat-\tfrac12\,\NNt\,\ethb\Nhat-\tfrac12\,\NNtbar\,\eth\Nhat\right]\\
      -\tfrac12\,\dd^{-1}\Nhat^2\left[\,\aaa\left\{\eth\ethb\Nhat-\BB\,\ethb\Nhat\right\}-
      \bb\left\{\ethb^2\!\Nhat-\tfrac12\,\AAAbar\,\ethb\Nhat-\tfrac12\,\CCbar\,\eth\Nhat\right\}+``cc"\right]\\
      -\Ascr\,\Nhat-\Bscr\,\Nhat^{\,3}=0\,,
  \end{multline}
  \begin{equation}
      \label{eq:phk}
      \partial_r\kk-\tfrac12\,\NNt\,\ethb\kk-\tfrac12\,\NNtbar\,\eth\kk-\tfrac12\,\Nhat\,\eth\KK+\ff=0\,,
  \end{equation}
  \begin{equation}
      \label{eq:phK}
      \partial_r\KK-\tfrac12\,\NNt\,\ethb\KK-\tfrac12\,\NNtbar\,\eth\KK-
      \tfrac{1}{2}\,\Nhat\,\dd^{-1}\Big\{\,\aaa(\eth\kkbar+\ethb\kk)-\bb\,\ethb\kkbar-\bbbar\,\eth\kk\,\Big\}+\FF=0\,,
  \end{equation}
  where the coefficients $\Ascr$, $\Bscr$, and the source terms $\ff$, $\FF$, in \eqref{eq:phN}, \eqref{eq:phk}, and \eqref{eq:phK}, are given as
  \begin{equation}\label{eq:phA}
  	\Ascr=\partial_r\Kstar-\tfrac12\,\NNt\,\ethb\Kstar-\tfrac12\NNtbar\,\eth\Kstar+\tfrac12\,\Big[\Kstar{}^2+\Kstar_{kl}\Kstar{}^{kl}\Big]\,,
  \end{equation}
  \begin{equation}\label{eq:phB}
  	\Bscr=-\tfrac12\,\Big[\Rhat+2\,\kkappa\,\KK+\tfrac12\,\KK^2-\dd^{-1}[2\,\aaa\,\kk\,\kkbar-\bb\,\kkbar^2-\bbbar\,\kk^2]-\Kc{}_{kl}\Kc{}^{kl}\Big]\,,
  \end{equation}
  \begin{multline}\label{eq:phf}
    \ff=-\tfrac12\,\Big[\kk\,\eth\NNtbar+\kkbar\,\eth\NNt\Big]-\left[\kkappa-\tfrac12\,\KK\right]\eth\Nhat+\Kstar\,\kk-\Nhat\Big[\eth\kkappa+q^i\dot{\nhat}{}^l\Kc{}_{li}-q^i\widehat{D}^l\Kc{}_{li}\Big]\,,
  \end{multline}
  \begin{align}\label{eq:phF}
    \FF=\tfrac{1}{4}\Nhat\,\dd^{-1}\Big\{2\,\aaa\,\BB\,\kkbar-\bb(\,\CCbar\,\kk+\AAAbar\,\kkbar)+``cc"\Big\}
    & -\dd^{-1}\Big[(\aaa\,\kkbar-\bbbar\,\kk)\,\eth\Nhat+``cc"\Big]\nonumber \\ & +\Big[\Kc_{ij}{\Kstar}{}^{ij}-\left(\kkappa-\tfrac12\,\KK\right)\Kstar\Big]\,
  \end{align}
  with $``cc"$ denoting the complex conjugate of the preceding terms, while the explicit form of the terms, such as $\Kstar_{ij}{\Kstar}{}^{ij}$, $\Kc_{ij}{\Kstar}{}^{ij}$, $\Kc_{ij}{\Kc}{}^{ij}$, $q^i\dot{\nhat}{}^l\Kc{}_{li}$, $q^i\widehat{D}^l\Kc{}_{li}$, can be found in \cite{Racz:2017krc}.
  
  A couple of comments are in order. First, recall that in Subsection \ref{subsec: 2+1}, using a $2+1$ decomposition, the pair $(h_{ab}, K_{ab})$ was replaced by the octuple $(\Nhat, \NN, \aaa, \bb; \kkappa, \KK, \kk, \Kcqq)$.  The above equations then suggest the following grouping of these variables. The unknown {\it constrained variables} are $(\Nhat, \KK, \kk)$. As an evolutionary system, \eqref{eq:phN}-\eqref{eq:phK} requires {\it initial data} corresponding to an arbitrary choice for the constrained variables $(\Nhat, \KK, \kk)$ on one of the $r=const$ level surfaces on $\Sigma$. Finally, the coefficients in \eqref{eq:phN}-\eqref{eq:phK} are determined by the {\it free data} $(\NN, \aaa, \bb, \kkappa, \Kcqq)$ or by terms which can be derived from them. In particular, $\Kstar$ is also completely determined by the free data. Note also that the evolutionary system \eqref{eq:phN}-\eqref{eq:phK} is guaranteed to have a unique smooth (local) solution for the {\it constrained fields} if a smooth choice has been made for the {\it free} and the {\it initial data} \cite{Racz:2015mfa}.

  \section{Asymptotically hyperboloidal initial data sets}
  \label{sec:hyper}
  
    Now we are ready to consider asymptotically hyperboloidal vacuum initial data sets by assuming some minimal asymptotic regularity of the constrained variables. This allows us to reformulate Proposition 2 of \cite{Beyer:2021kmi} as follows.
    
    \begin{proposition}\label{prop2}
    	Choose a generic set of free data $(\NN,\aaa,\bb, \kkappa,\Kcqq)$ on $\Sigma$ that satisfies the falloff conditions given by the equations  \eqref{eq:falloff1}-\eqref{eq:falloff2} in Proposition \ref{prop1}. Suppose that the asymptotic limit $\kkappa_0=\lim_{r\rightarrow\infty}\kkappa$ of $\kkappa$ is a strictly positive smooth function on $\partial\widetilde{\Sigma}$.
    	Assume also that on $\Sigma$ $(\Nhat,\KK,\kk)$ are smooth solutions of the parabolic-hyperbolic form of the constraints \eqref{eq:phN}-\eqref{eq:phK}, whose coefficients are derived from the chosen free data $(\NN,\aaa,\bb, \kkappa,\Kcqq)$. Finally, assume that the solution has the following asymptotic properties
    	\begin{itemize}
    		\item[(i)] \vskip-0.3cm 
    		$\Nhat$ is strictly positive and has a second order asymptotic radial expansion\footnote{A variable $f$ has an asymptotic radial expansion of order $n^{th}$ if near infinity $f$ can be written in the form
			\begin{equation}
				f=f_0+f_1\cdot r^{-1}+\dots+f_{n-1}\cdot r^{-(n-1)}+\mathscr{O}(r^{-n})
			\end{equation}
			where the coefficients $f_0, f_1,\dots,f_{n-1}$ are smooth functions on $\partial\widetilde{\Sigma}$, i.e. they are independent of $r$. Note that the indices used here refer to the power of $\omega$, not the power of $r$.},
    		\item[(ii)] \vskip-0.15cm 
    		$\KK$ and $\kk$ have a first-order asymptotic radial expansion. 
    	\end{itemize}
    	\vskip-0.3cm  
    	Then
    	\begin{equation}
    		\Nhat_0=0\,, \ \ \Nhat_1={\kkappa_0}^{-1}\,,\ \ \KK_0=2{\kkappa_0}\,,\ \ \kk_0={\kkappa_0}^{-1}\eth{\kkappa_0}\,,
    	\end{equation}
    	and, incidentally, the resulting vacuum initial data $(\Nhat,\KK,\kk; \NN,\aaa,\bb, \kkappa,\Kcqq)$ on $\Sigma$ is asymptotically hyperboloidal.  
    \end{proposition}
    Note that the background data $(\NN, \aaa, \bb, \kkappa, \Kcqq)$ can be generic, in the sense that it does not have to be related to the initial data induced on a time slice by a given solution of the vacuum Einstein equations. In contrast, the resulting initial data $(\Nhat, \KK, \kk; \NN, \aaa, \bb, \kkappa, \Kcqq)$ on $\Sigma$ is undoubtedly an asymptotically hyperboloidal solution of the vacuum constraint equations. It should also be noted that while in Proposition 2 of \cite{Beyer:2021kmi} $\kkappa_0$ was assumed to be a constant on $\partial\widetilde{\Sigma}$, in the above results $\kkappa_0$ may have an angular dependence, since we allowed it to be an arbitrary  smooth strictly positive function on $\partial\widetilde{\Sigma}$. Note also that the conditions in Proposition \ref{prop2} are not yet sufficient to exclude the occurrence of logarithmic singularities, as will be verified by the proof of Theorem \ref{main theorem 1} below.
    
    \medskip
    
    Before turning to our main results it is high time to relate the setup used in this paper to that used in \cite{Andersson:1992yk,Andersson:1993we,Andersson:1994ng,Andersson:1996xd}. Before jumping into the technicalities, note that the main motivation in the studies of Andersson and Chru\'sciel \cite{Andersson:1992yk,Andersson:1993we,Andersson:1994ng,Andersson:1996xd} was to identify the differential properties of potential asymptotically hyperboloidal solutions of \eqref{eq:Ham} and \eqref{eq:Mom} within the elliptic setup. In contrast, we use the parabolic-hyperbolic formulation of \eqref{eq:Ham} and \eqref{eq:Mom} and our main motivation is to first identify those conditions which guarantee the existence of well-defined and finite Bondi mass and angular momentum, and then to study the differentiability properties of the selected class of asymptotically hyperboloidal initial data configurations.
     
    In formulating the main result of \cite{Andersson:1992yk,Andersson:1993we,Andersson:1994ng,Andersson:1996xd} we need to introduce the following notation. Denote by $\tildeon{K}{}_{ab}$ the extrinsic curvature of $\widetilde{\Sigma}$ with respect to the conformally rescaled three-metric  $\tildeon{h}{}_{ab}=\Omega^2{h}{}_{ab}$, and with respect to the vector $\tildeon{n}{}^a=\Omega^{-1}n^a$ normal to $\widetilde\Sigma$. Similarly, $\tildeon{\widehat{K}}{}_{ab}$ denotes the extrinsic curvature of the $r=const$ level sets within $\widetilde{\Sigma}$, i.e., with respect to the conformally rescaled metric  $\tildeon{\widehat{\gamma}}{}_{ab}=\omega^2\widehat{\gamma}{}_{ab}$, and with respect to the normal $\tildeon{\widehat{n}}{}^a=\omega^{-1}\widehat{n}^a$.  Denote by ${\tildeon{\kkappa}}, {\tildeon{\kk}}{}_{a}, {\tildeon{\KK}}{}_{ab}$ the scalar, vector, and tensor projections of $\tildeon{K}{}_{ab}$ obtained by a decomposition analogous to that applied in Subsection \ref{subsec: 2+1} above.
    Also denote by ${\tildeon{\KK}}{}{}^\circ_{ab}$ and $\tildeon{\widehat{K}}{}^\circ_{ab}$ the trace-free parts of ${\tildeon{\KK}}{}_{ab}=\tildeon{\widehat{\gamma}}{}_{a}{}^e\tildeon{\widehat{\gamma}}{}_{b}{}^f\tildeon{{K}}{}_{ef}$ and $\tildeon{\widehat{K}}{}_{ab}$, respectively. Finally recall that the generic solution to the momentum constraint \cite{Andersson:1992yk,Andersson:1993we,Andersson:1994ng,Andersson:1996xd} was found to possess the form
    \begin{equation}
    	\tildeon{K}{}_{ab}=\tildeon{K}{}^{[C^\infty]}_{ab}+\Omega^2\log\Omega\cdot\tildeon{K}{}_{ab}^{[log]}\,,
    \end{equation}
    where $\tildeon{K}{}^{[C^\infty]}_{ab}$ is smooth all over $\widetilde{\Sigma}$, while $\tildeon{K}{}_{ab}^{[log]}$ has a polyhomogeneous expansion, whence it is smooth only over $\Sigma$ and extends merely continuously to $\partial\widetilde{\Sigma}$.
    
    Before presenting the main result of Andersson and Chru\'sciel, note that in \cite{Andersson:1993we,Andersson:1994ng} they linked the geometry of the boundary of the initial data surface with the geometry of the resulting spacetime. Using the variables introduced above, we can then reformulate it as follows.
    
    \begin{proposition}[Andersson $\&$ Chru\'sciel \cite{Andersson:1993we,Andersson:1994ng}]\label{prop3}
    	Consider a generic asymptotically hyperboloidal initial data set $(\Nhat,\KK,\kk; \NN,\aaa,\bb, \kkappa,\Kcqq)$ on $\Sigma$ that satisfies the assumptions used in Proposition \ref{prop1}. Suppose that the following three relations
    	\begin{equation}\label{eq: conds-Prop3}
    		\tildeon{K}{}_{rr}^{[log]}=0\,, \quad \tildeon{K}{}_{ra}^{[log]}=0\,, \quad   {\tildeon{\KK}}{}{}^\circ_{ab} - \tildeon{\widehat{K}}{}^\circ_{ab}=0
    	\end{equation}
        hold simultaneously on $\partial\widetilde{\Sigma}$.
    	Then there exists a Cauchy development of such generic asymptotically hyperboloidal initial data that admits a smooth conformal boundary.
    \end{proposition}
    
    It is important to reformulate the conditions in \eqref{eq: conds-Prop3} in terms of the variables we use in our investigations. To uncover the necessary connections, we need to prove the following lemma.
    
    \begin{lemma}\label{lemma1}
    	Let $(\Nhat,\KK,\kk; \NN,\aaa,\bb, \kkappa,\Kcqq)$ be a generic asymptotically hyperboloidal initial data set on $\Sigma$ that satisfies the assumptions used in Proposition \ref{prop1}. Assume also that the background fields $(\NN,\aaa,\bb,\kkappa,\Kcqq)$ are smooth on $\Sigma$, they extend smoothly to $\partial\widetilde{\Sigma}$, and that the constrained fields $(\Nhat,\KK,\kk)$ are smooth solutions of the parabolic-hyperbolic form of the constraints \eqref{eq:phN}-\eqref{eq:phK} over $\Sigma$ and their asymptotic expansion can be given by the most general polyhomogeneous expressions (i.e., there exists a sequence $\{\mathcal{N}_j\}_{j=1}^{\infty}$ and also expansion coefficients  $\Nhat_{i}, \Nhat_{i,j}^{[log]},\KK_{i}, \KK_{i,j}^{[log]}, \kk_{i}, \kk_{i,j}^{[log]}$ that are smooth functions on $\mathbb{S}^2$) so that 
    	\begin{align}
    		&\Nhat=\sum_{i=1}^{\infty}\omega^i\,\big[\,\Nhat_i+\sum_{j=1}^{\mathcal{N}_j}\Nhat_{i,j}^{[log]}\log^j\omega\,\big]\,,\label{eq:Nhrseries-0}\\
    		&\KK=\KK_0+\sum_{i=1}^{\infty}\omega^i\,\big[\,\KK_i+\sum_{j=1}^{\mathcal{N}_j}\KK_{i,j}^{[log]}\log^j\omega\,\big]\,,\label{eq:Krseries-0}\\\
    		&\kk=\kk_0+\sum_{i=1}^{\infty}\omega^i\,\big[\,\kk_i+\sum_{j=1}^{\mathcal{N}_j}\kk_{i,j}^{[log]}\log^j\omega\,\big]\,.\label{eq:krseries-0}
    	\end{align}
    	Then the following asymptotic relations hold
    	\begin{align}
    		 &\big({\tildeon{\KK}}{}{}^\circ_{ab} - \tildeon{\widehat{K}}{}^\circ_{ab}\big)q^aq^b = {} \Kcqq{}_{^{(-1)}}+{\bb{}_{^{(-1)}}}\,\bigg[{\Nhat_{1}+\sum_{j=1}^{\mathcal{N}_j}\Nhat_{1,j}^{[log]}\log^j\omega}\bigg]^{-1} + \mathscr{O}(\omega)\,,\label{eq:shear} \\
    		\tildeon{K}{}_{rr}^{[log]} = {} & - \mathscr{L}_n \Omega\cdot\bigg[\,2\,\Nhat_{1,0}\sum_{j=1}^{\mathcal{N}_j}\Nhat_{1,j}^{[log]}\log^{j-1}\omega+\sum_{l=1,j=1}^{\mathcal{N}_l,\mathcal{N}_j}\Nhat_{1,l}^{[log]}\Nhat_{1,j}^{[log]}\log^{l+j-1}\omega\bigg] + \mathscr{O}(\omega)\,, \label{eq:lKrr} \\
    	& 	\tildeon{K}{}_{ra}^{[log]} = {} \tfrac12\,\big(\kk_0 \,{\bar q}{}_a + {\bar\kk}{}_0\, q_a\big)\cdot\bigg[\sum_{j=1}^{\mathcal{N}_j}\Nhat_{1,j}^{[log]}\log^{j-1}\omega\bigg] + \mathscr{O}(\omega) \,. \label{eq:lKra}
    	\end{align}
    Since $q^a$ is a complex dyad, the algebraic content of ${\tildeon{\KK}}{}{}^\circ_{ab} - \tildeon{\widehat{K}}{}^\circ_{ab}$ and that of $\big({\tildeon{\KK}}{}{}^\circ_{ab} - \tildeon{\widehat{K}}{}^\circ_{ab}\big)q^aq^b$ are equivalent (for a verification see \eqref{eq: intTqq}).
    \end{lemma}
    \begin{proof}
    First recall that using the rescaling of $\tildeon{n}{}^a=\Omega^{-1}n^a$ and $\tildeon{h}{}_{ab}=\Omega^2{h}{}_{ab}$, and that of $\tildeon{\widehat{n}}{}^a=\omega^{-1}\widehat{n}^a$ and $\tildeon{\widehat{\gamma}}{}_{ab}=\omega^2\widehat{\gamma}{}_{ab}$, it is straightforward to show (see also Appendix \ref{appendix:decomp}) that
    \begin{equation}
    	\tildeon{K}{}_{ab} = \Omega \,{K}{}_{ab} - \mathscr{L}_n \Omega\cdot {h}_{ab} \quad {\rm and} \quad \tildeon{\widehat{K}}{}_{ab} = \omega \,{\widehat{K}}{}_{ab} - \mathscr{L}_{\widehat{n}} \omega\cdot {\widehat{\gamma}}_{ab}\,,
    \end{equation}
    which immediately implies 
    \begin{equation}
    	\tildeon{\KK}{}^\circ_{ab} = \omega \,\circon{\KK}{}_{ab}  \quad {\rm and} \quad \tildeon{\widehat{K}}{}^\circ_{ab} = \omega\, \circon{\widehat{K}}{}_{ab}\,,
    \end{equation}
    and in turn verifies 
    \begin{equation}\label{eq: externals}
    	\tildeon{\KK}{}^\circ_{ab} - \tildeon{\widehat{K}}{}^\circ_{ab}= \omega \,\big[\circon{\KK}{}_{ab}-\circon{\widehat{K}}{}_{ab}\big]  \,.
    \end{equation}
    Note also that due to \eqref{eq:falloff1}-\eqref{eq:falloff2} and our assumption on the smoothness of free data on $\widetilde{\Sigma}$, the following asymptotic relations hold 
	\begin{align}  
		&\NN = \NN_1 \,\omega + \NN_2 \, \omega^{2} + \mathscr{O}(\omega^{3})\,,\label{eq:falloffT12NNom}\\\
		&\aaa =\omega^{-2} + \aaa{}_{^{(-1)}}\,\omega^{-1} + \aaa_0 + \aaa_1\, \omega + \aaa_2\,\omega^{2} + \mathscr{O}(\omega^{3})\,,\\
		&\bb =\bb{}_{^{(-1)}}\,\omega^{-1} +\bb_0 + \bb_1 \,\omega + \bb_2 \,\omega^{2} + \mathscr{O}(\omega^{3})\,, \\
		&\kkappa = \kkappa_0 + \kkappa_{1}\, \omega + \kkappa_2\, \omega^{2} + \mathscr{O}(\omega^{3}) \,,\\
		&\Kcqq = \Kcqq{}_{^{(-1)}}\,\omega^{-1} + \Kcqq{}_0+ \Kcqq{}_1 \,\omega+ \Kcqq{}_2\, \omega^{2}+ \mathscr{O}(\omega^{3}) \,. \label{eq:falloffT12Kcqqom}
	\end{align}
    Using \eqref{inthatextcurv_qq} and substituting \eqref{eq:falloffT12NNom}-\eqref{eq:falloffT12Kcqqom} into the asymptotic expansion of $\circon{\KK}{}_{qq} - \circon{\widehat{K}}{}_{qq}$, we get that 
    \begin{equation}
    	\Kcqq-\interior{\widehat{K}}{}_{qq}=\bigg(\Kcqq{}_{^{(-1)}}+\bb{}_{^{(-1)}}\,\bigg[{\Nhat_{1}+\sum_{j=1}^{\mathcal{N}_j}\Nhat_{1,j}^{[log]}\log^j\omega}\bigg]^{-1}\bigg)\,\omega^{-1} + \mathscr{O}(\omega^0)\,, 
    \end{equation} 
    which, in virtue of \eqref{eq: externals}, implies that
    \begin{equation}\label{eq:shear0}
    	\big({\tildeon{\KK}}{}{}^\circ_{ab} - \tildeon{\widehat{K}}{}^\circ_{ab}\big)\,q^aq^b= \Kcqq{}_{^{(-1)}}+\bb{}_{^{(-1)}}\,\bigg[{\Nhat_{1}+\sum_{j=1}^{\mathcal{N}_j}\Nhat_{1,j}^{[log]}\log^j\omega}\bigg]^{-1} + \mathscr{O}(\omega) \,
    \end{equation}
    holds close to $\partial\widetilde{\Sigma}$ which in turn verifies \eqref{eq:shear}.
    
    \medskip
    
    Then, using the $3+1$ decomposition of $\tildeon{K}{}_{ab}$ and $K_{ab}$, by the relations \eqref{appendix:decomp}, we get
    \begin{align}
    	\tildeon{K}{}_{rr} & = \Omega \,{K}{}_{rr} - \mathscr{L}_n \Omega\cdot {h}_{rr}=\Omega \,[\kkappa \widehat{N}^2+ 2\,\widehat{N}\,\kk_e\Nhat^e+\KK_{ab}\Nhat^a\Nhat^b] - \mathscr{L}_n \Omega\cdot[\widehat{N}^2+\Nhat_e\Nhat^e]
    	\,, \label{eq:tKrr}\\
    	\tildeon{K}{}_{ra} & =\Omega \,{K}{}_{ra} - \mathscr{L}_n \Omega\cdot {h}_{ra}=\Omega \,[\widehat{N}\,\kk_a+(\circon{\KK}_{ae}+\tfrac12\KK\gammahat_{ae})\Nhat^e] - \mathscr{L}_n \Omega\cdot\Nhat_a\,,\label{eq:tKra}
    \end{align}
    where $\mathscr{L}_n \Omega$ must have non-vanishing regular asymptotic limit since we have that $\mathscr{L}_n \Omega={N}^{-1}\partial_{T^a-N^a} \Omega\sim{N}^{-1}\partial_{T} \Omega$, where $n^a=N^{-1}(T^a-N^a)$ is the unit normal to $\Sigma$, $T^a$, $N$, and $N^a$ are the evolution vector field, the four-dimensional lapse and the shift, respectively. The non-vanishing regular asymptotic limit of $\mathscr{L}_n \Omega$ then exists as $N$ and $\partial_{T} \Omega$ must have finite asymptotic limits, and that $N^a$ asymptotically tends to zero.
    Then, substituting the relations $\kk_a=\tfrac12\,\big[ \kk \,{\bar q}{}_a + {\bar\kk}\, q_a\big]$, $\Nhat_a=\tfrac12\,\big[ \NN \,{\bar q}{}_a + {\bar{\NN}}\, q_a\big]$, $\Nhat^a=\tfrac12\,\big[ \tildeon{\NN} \, {\bar q}{}^a + {\bar{\tildeon{\NN}}}\, q^a\big]=\tfrac12\,\dd^{-1}\big[(\aaa\NN-\bb\NNbar)\,{\bar q}{}^a +(\aaa\NNbar-\bbbar\NN)\, q^a\big]$, along with \eqref{appendix:gammaNN} - \eqref{appendix:KNN} and the asymptotic expansions \eqref{eq:Nhrseries-0}-\eqref{eq:krseries-0} into \eqref{eq:tKrr} and \eqref{eq:tKra} we get by a straightforward algebraic manipulation that \eqref{eq:lKrr} and \eqref{eq:lKra} also hold, which completes our proof.
    \end{proof}

     The vanishing of the leading order contributions in \eqref{eq:lKrr} requires the vanishing of each $\Nhat_{1,j}^{[log]}$ coefficient separately. Then \eqref{eq:lKra} is automatically fulfilled. Finally the vanishing of the leading order contributions in \eqref{eq:shear} together with $\Nhat_{1,j}^{[log]}=0$ gives the relation $\Kcqq{}_{^{(-1)}}+{\bb{}_{^{(-1)}}}\Nhat_{1}{}^{-1}=0$. Accordingly, the three conditions \eqref{eq: conds-Prop3} in Proposition \ref{prop3} are equivalent to the vanishing of logarithmic terms $\Nhat_{1,j}^{[log]}=0$, and  the relation $\Kcqq{}_{^{(-1)}}+{\bb{}_{^{(-1)}}}\Nhat_{1}{}^{-1}=0$.
    
    \section{The main results}\label{sec:main}
  
    The purpose of this section is to formulate and outline the proofs of the main results of this paper. Most importantly, this subsection aims to clarify several issues that were not sufficiently discussed in \cite{Beyer:2021kmi}. We will assume that there are smooth solutions $(\Nhat,\KK,\kk)$ to the constraint equations on $\Sigma$ that extend only continuously to $\partial\widetilde{\Sigma}$. Accordingly, we allow them to have the most general polyhomogeneous expansions as given by \eqref{eq:Nhrseries-0}-\eqref{eq:krseries-0} in a neighborhood of $\partial\widetilde{\Sigma}$. We also assume that the choice of free data $(\NN,\aaa,\bb,\kkappa,\Kcqq)$ corresponds to the generic asymptotically hyperboloidal initial data sets that can be given by the asymptotic expansions \eqref{eq:falloffT12NNom}-\eqref{eq:falloffT12Kcqqom}. Within this framework, we then examine what kind of restrictions on the coefficients of the asymptotic expansions arise from the existence of well-defined finite Bondi mass and angular momentum, and also from the assumption that the parabolic-hyperbolic form of the constraint equations holds. It turns out that all the conditions implicitly used in Proposition 3 of \cite{Beyer:2021kmi}, concerning the absence of logarithmic terms follow directly from the existence of well-defined finite Bondi mass and angular momentum and from the assumption that the parabolic-hyperbolic form of the constraint equations has smooth solutions on the whole of $\Sigma$. In addition, we also explain part of the implicitly used falloff conditions imposed in Proposition 3 of \cite{Beyer:2021kmi} on $\aaa,\bb$ and $\Kcqq$, we also show that these were overly restrictive. It should also be noted that both sets of these conditions were imposed in \cite{Beyer:2021kmi} without clearly stating them, and also without openly discussing the reason for their application.
    It is also important to emphasize, as we will show below in Corollary \ref{corollary}, that, surprisingly, the existence of well-defined Bondi mass within the parabolic-hyperbolic setup selects a subspace of asymptotically hyperboloidal initial data configurations which admit Cauchy developments, as concluded when studying solutions to the elliptic form of the constraints in \cite{Andersson:1992yk,Andersson:1993we,Andersson:1994ng,Andersson:1996xd}, with smooth, hence logarithmic singularity-free null infinity.
    
    \subsection{Well-defined Bondi mass and angular momentum}
    \label{sec:main-bondi}
    
    We are now ready to formulate our first main result.
    
        \begin{theorem}\label{main theorem 1}
    	Choose a generic asymptotically hyperboloidal set of free data $(\NN,\aaa,\bb, \kkappa,\Kcqq)$ on $\Sigma$ which satisfies the falloff conditions given by the equations \eqref{eq:falloff1}-\eqref{eq:falloff2} in Proposition \ref{prop1}. 
    	Suppose that on $\Sigma$ $(\Nhat,\KK,\kk)$ are smooth solutions of the parabolic-hyperbolic form of the constraints \eqref{eq:phN}-\eqref{eq:phK}, whose coefficients are derived from the chosen free data, and which allow the most general polyhomogeneous expansion of $(\Nhat,\KK,\kk)$ as given by \eqref{eq:Nhrseries-0}-\eqref{eq:krseries-0} in a neighborhood of $\partial\widetilde{\Sigma}$. The asymptotically hyperboloidal initial data set under consideration admits well-defined Bondi mass and angular momentum if and only if all coefficients of the logarithmic terms in \eqref{eq:Nhrseries-0}-\eqref{eq:krseries-0} vanish up to order four and three for $\Nhat,\KK$ and $\kk$, respectively, and, in addition,
    	\begin{equation}\label{eq:alg-theorem}
    		\Kcqq{}_{^{(-1)}}=0\,, \quad \bb{}_{^{(-1)}}=0\,, \quad \kkappa_1=0\,.
    	\end{equation}
    \end{theorem}
    \begin{proof}
    	An initial data set $(\Nhat,\KK,\kk; \NN,\aaa,\bb, \kkappa,\Kcqq)$ is considered physically adequate if it admits well-defined and thus finite Bondi mass and angular momentum. We begin the proof by proving two lemmas that will be used to verify the statement of our theorem.
    	
    	\begin{lemma}\label{lemma: Bondi-mass}
    	The Bondi mass can be finite, and thus well-defined, if and only if for the expansion coefficients in \eqref{eq:Nhrseries-0} and \eqref{eq:Krseries-0} the following relations hold
    	\begin{align}
    		\Nhat_{1} & = 2 \,\KK_0^{-1}\, \label{eq: lemmaBM1} \\ \Nhat_{2} & = -\big[\,\aaa{}_{^{(-1)}}\, \KK_0 + 2\, \KK_1\,\big]\,\KK_0^{-2}\,  \\
    		\Nhat_{3} & = \bigg(2\,(\KK_1^2-2) + \KK_0\,\big[\,\aaa{}_{^{(-1)}} \KK_1 -2\,\KK_2\big] \nonumber \\ & \hskip2.5cm - \KK_0^2\,\big(\,2\,\aaa_0 - \aaa{}_{^{(-1)}}^2 - \bb{}_{^{(-1)}}\overline{\bb{}_{^{(-1)}}}  + \tfrac12\,[\,\eth\overline{\NN}_1+\overline{\eth}\NN_1]\big)\bigg)\,\KK_0^{-3} \,,\label{eq: lemmaBM3} 
    	\end{align}
    	and also for all $j=1,2,\dots,\mathcal{N}_j$
    	\begin{align}
    		\Nhat_{1,j}^{[log]} & = \Nhat_{2,j}^{[log]}=\KK_{1,j}^{[log]}=0\,, \label{eq: lemmaBM5} \\
    		\KK_{2, i}^{[log]} & = \KK_0\,\big(\,2 \, \KK_{3, i}^{[log]} + \KK_0^2\, \Nhat_{4, i}^{[log]}\,\big) \cdot \big[\,\aaa{}_{^{(-1)}}\, \KK_0 + 4\, \KK_1\,\big]^{-1}\,,\label{eq: lemmaBM6}\\
    		\Nhat_{3, i}^{[log]} & = -2\,\big(\,2 \, \KK_{3, i}^{[log]} + \KK_0^2\, \Nhat_{4, i}^{[log]}\,\big) \cdot \big(\KK_0\,\big[\,\aaa{}_{^{(-1)}}\, \KK_0 + 4\, \KK_1\,\big]\big)^{-1}\,.\label{eq: lemmaBM7}
    	\end{align}
    	\end{lemma}
    	\begin{proof}\hskip-0.35cm({\bf of Lemma \ref{lemma: Bondi-mass}})
    	To prove our lemma first, we recall that the Bondi mass can be given as the $r\rightarrow \infty$ limit of the Hawking mass that is evaluated on $\mathscr{S}_r$ level surfaces as
    	\begin{equation}
    		m_H=\sqrt{\frac{\mathcal{A}}{16\pi}}\left(1+\frac{1}{16\pi}\int_{\mathscr{S}_r}\Theta^{(+)}\Theta^{(-)}\epshat\right)\,,
    		\label{eq:Hexp}
    	\end{equation}
    	where $\epshat$ is the volume element associated with $\gammahat_{ab}$, $\mathcal{A}=\int_{\mathscr{S}_r}\epshat$ is the area of the ${\mathscr{S}_r}$ level surfaces, and where
    	\begin{equation}\label{eq: thetaPM}
    		\Theta^{(\pm)}=\KK\pm\Kstar\Nhat^{-1}
    	\end{equation}
    	denote the null expansions with respect to $n^{(\pm)}_a=n_a\pm\widehat{n}_a$, where $n_a$ is the timelike normal to $\Sigma$ and $\widehat{n}_a$ is the normal to the $r=const$ level surfaces within $\Sigma$, and $\Kstar=\Nhat\Kh$.
    	
    	Since we are interested in the $r\rightarrow \infty$ limit of the Hawking mass, it is sufficient to study its asymptotic behavior. For the asymptotic expansion of the prefactor $\sqrt{\mathcal{A}/16\,\pi}$ the relation
    	\begin{equation}
    		\sqrt{\mathcal{A}/16\,\pi}=r/2+\OO{0}
    	\end{equation}
    	holds, the Hawking mass $m_H$ cannot tend to a finite value in the asymptotic limit, unless for the integral term the relation
    	\begin{equation}\label{eq:intth+th-}		
    		\int\Theta^{(+)}\Theta^{(-)}\,\epshat=\int\Theta^{(+)}\Theta^{(-)}\,\sqrt{\dd}\,\boldsymbol{\epsilon}_q=-16\, \pi+\OO{-1}
    	\end{equation}
    	where $\dd$ is the ratio of the determinant of $\gammahat_{ab}$ to that of the unit sphere metric $q_{ab}$, and $\boldsymbol{\epsilon}_q$ is the volume element of the unit sphere, which, in standard spherical coordinates, is $\sin\vartheta\,\mathrm{d}\vartheta\wedge\mathrm{d}\varphi$.
    	
    	In the next step we have to study the asymptotic expansion of the integrand $\Theta^{(+)}\Theta^{(-)}\sqrt{\dd}$, which due to \eqref{eq: thetaPM} includes $\Nhat, \KK, \Kstar, \dd$ and, since we are interested in a well-defined and finite Bondi mass, can  be formally written without loss of generality as
    	\begin{equation}\label{eq: exp-thetaPM}
    		\Theta^{(+)}\Theta^{(-)}\sqrt{\dd} = f_{-2}\,r^2 + f_{-1}\,r + f_{0} +f_{1}\,r^{-1} + \mathscr{O}(r^{-2})\,.
    	\end{equation}
    	Next we have to determine the involved critical coefficients $f_{-2},f_{-1},f_{0},f_{1}$ in \eqref{eq: exp-thetaPM}. Before doing so, note that the above discussion implies that the necessary and sufficient condition for the existence of a well-defined Bondi mass is the vanishing of $f_{-2}$ and $f_{-1}$ and that $f_{0}=-4$, and also that $f_{1}$ is finite.
    	
    	Substituting the asymptotic expansions of the free data given by  \eqref{eq:falloffT12NNom}-\eqref{eq:falloffT12Kcqqom} and those of $\Nhat$ and $\KK$ given by \eqref{eq:Nhrseries-0} and \eqref{eq:Krseries-0} into the asymptotic expansion of the integrand $\Theta^{(+)}\Theta^{(-)}\sqrt{\dd}$ which is given by \eqref{eq: thetaPM}, and includes $\Nhat, \KK, \Kstar, \dd$, we get that the vanishing of $f_{-2}$ and $f_{-1}$ and that $f_{0}=-4$, and also that $f_{1}$ is finite requires \eqref{eq: lemmaBM1}-\eqref{eq: lemmaBM3} to hold, and that the critical coefficients in \eqref{eq: exp-thetaPM} 
    	cannot be finite unless the  relations \eqref{eq: lemmaBM5}-\eqref{eq: lemmaBM7} hold for the coefficients of the logarithmic terms in the asymptotic expansions of $\Nhat$ and $\KK$. 
    	\end{proof}
    	
    	\begin{lemma}\label{lemma: Bondi-ang}
    	The Bondi angular momentum cannot be finite, and thus well-defined, unless for all $j=1,2,\dots,\mathcal{N}_j$
    	\begin{equation}\label{eq: lemmaBJ} 
    		\kk_{1,j}^{[log]} = \kk_{2,j}^{[log]}=0\,,  
    	\end{equation}
    	hold for the coefficients of the logarithmic terms appearing in the asymptotic expansion \eqref{eq:krseries-0}. 	
    	\end{lemma}
    	\begin{proof}\hskip-0.35cm({\bf of Lemma \ref{lemma: Bondi-ang}})
    	As it was shown by one of the present authors IR in \cite{Racz:2024giv}, the quasi-local angular momentum determined by an axial vector field $\phi^a$ (i.e. a vector field with closed orbits (with period $2\,\pi$, and with two poles), and with vanishing divergence with respect to the induced connection $\widehat{D}_a$) tangent to a $r=const$ level surface is given by
    	\begin{equation}\label{eq: qlam}
    		J[\phi] = - (8\, \pi)^{-1}\int_{{\mathscr{S}}_r} \phi^a{ \rm \bf k}_a \, \widehat{{\boldsymbol{\epsilon}}} \,.
    	\end{equation}
    	It was also shown in \cite{Racz:2024giv} that the generic form of an axial vector field $\phi^a$ can be given as
    	\begin{equation}\label{eq: axial}
    		{\phi}{}^a = {\underline{\sqrt{{\widehat{\gamma}}/{q}}}}_{\,[{\tiny{\interior{\phi}}}]}\, \left[\sqrt{{q}/{\widehat{\gamma}}\,}\, \interior{\phi}{}^a\right],
    	\end{equation}
    	where $\interior{\phi}{}^a$ is an axial Killing vector field of the unit sphere metric $q_{ab}$, $\widehat{\gamma}$ and $q$ denote the determinants of the metrics $\widehat{\gamma}_{ab}$ and $q_{ab}$ on $\mathscr{S}$, respectively, and the averaging factor
    	\begin{equation}
    		{\underline{\sqrt{{\widehat{\gamma}}/{q}}}}_{\,[{\tiny{\interior{\phi}}}]} = \frac1{2\,\pi} \int_{0}^{2\,\pi}\sqrt{{\widehat{\gamma}}/{q}}  \,\,d\interior{\varphi}\,,
    	\end{equation}
    	was also used, which is constant along the integral curves of the axial Killing vector field $\interior{\phi}{}^a$, and $\interior{\varphi}$ denotes the corresponding $2\,\pi$-periodic axial coordinate.
    	
    	Putting all the above observations together, we get 
    	\begin{equation}\label{eq: qlam1}
    		J[\phi] = - (8\, \pi)^{-1}\int_{{\mathscr{S}}_r} {\underline{\sqrt{\dd}}}_{\,[{\tiny{\interior{\phi}}}]}\,(\interior{\phi}{}^a{ \rm \bf k}_a) \, \boldsymbol{\epsilon}_q \,,
    	\end{equation}
    	where $\boldsymbol{\epsilon}_q$ is the volume element of the unit sphere metric $q_{ab}$. It is well-known that there are many ways to fix a unit sphere metric on ${{\mathscr{S}}_r}$. Nevertheless, it is possible to reduce the corresponding freedom considerably by restricting our attention to centre-of-mass unit sphere reference systems (for more details, see, e.g., \cite{Klainerman:2019uaa,Racz:2024giv}) selected by the conditions
    	\begin{equation}\int_{{\mathscr{S}}_r} \sqrt{\dd}\, \vec{x}\, \boldsymbol{\epsilon}_q = \vec{0}\,,
    	\end{equation}
    	where $\vec{x}=(\cos\varphi\sin\vartheta,\sin\varphi\sin\vartheta,\cos\vartheta)$ and $\vec{0}$ denotes the three-dimensional zero vector. These centre-of-mass unit sphere reference systems still allow a three-parameter family subsystem determined up to three-dimensional rotations of the unit sphere in $\mathbb{R}^3$.
    	
    	It also follows from the  analysis in \cite{Racz:2024giv} that the quasi-local angular momentum expressions $J[\phi]$ determined by a distinguished axial vector field ${\phi}{}^a$ can be finite only if each of the three principal angular momentum expressions, even though $q_{ab}$ may be boosted with respect to a centre-of-mass unit sphere reference system,
    	\begin{equation}\label{eq:qlam2}
    		J[\phi_{{}^{{}^{{}_{{}^{(i)}}}}}\hskip-0.06cm] = - (8\, \pi)^{-1}\int_{{\mathscr{S}}_r} {\underline{\sqrt{\dd}}}_{\,[{\tiny{\interior{\phi}_{{}^{{}^{{}_{{}^{(i)}}}}}}}\hskip-0.06cm]}\,(\interior{\phi}{}^a_{{}^{{}^{{}^{{}^{(i)}}}}}\hskip-0.06cm{ \rm \bf k}_a) \, \boldsymbol{\epsilon}_q \,.
    	\end{equation} 
    	are finite, where $\interior{\phi}{}^a_{{}^{{}^{{}^{{}^{(i)}}}}}\hskip-0.06cm$, $(i=1,2,3)$, denote the three axial Killing vector fields of $({\mathscr{S}}_r,q_{ab})$
    	\begin{align}\label{eq:rotation-unit}
    		\interior{\phi}{}^a_{{}^{{}^{{}^{{}^{(1)}}}}}\hskip-0.06cm & = -\sin\varphi \,(\partial_\vartheta)^a - \cot \vartheta\,\cos\varphi \,(\partial_\varphi)^a \\
    		\interior{\phi}{}^a_{{}^{{}^{{}^{{}^{(2)}}}}}\hskip-0.06cm & = \cos\varphi \,(\partial_\vartheta)^a - \cot \vartheta\,\sin\varphi\,(\partial_\varphi)^a \\	
    		\interior{\phi}{}^a_{{}^{{}^{{}^{{}^{(3)}}}}}\hskip-0.06cm & = (\partial_\varphi)^a\,,
    	\end{align}
    	corresponding to the generators of the rotations about the $x,y,z$-axis of a unit sphere on $\mathbb{R}^3$. Thus the Bondi angular momentum can be finite only if the limits of the quasi-local angular momentum expressions $J[\phi_{{}^{{}^{{}_{{}^{(i)}}}}}\hskip-0.06cm]$ are finite. 
    	
    	By \eqref{eq:qlam2} and \eqref{eq:rotation-unit}, the quasi-local angular momentum expressions $J[\phi_{{}^{{}^{{}_{{}^{(i)}}}}}\hskip-0.06cm]$ include the $\vartheta$ and $\varphi$ components of $\kk_a$, which, by choosing the dyad through
    	\begin{equation}
    		q^a=(\partial_\vartheta)^a+\frac{\mathbbm{i}}{\sin\vartheta}(\partial_\varphi)^a\,,
    	\end{equation}
    	can be written as  
    	\begin{equation}
    		\kk_\vartheta=\tfrac12\,[\kk_aq^a+\kk_a\bar{q}^a]=\tfrac12\,[\kk+\overline{\kk}]\quad {\rm and} \quad \kk_\varphi=\tfrac12\,[\kk_aq^a-\kk_a\bar{q}^a]=\tfrac{\mathbbm{i}}{2}\,\sin\vartheta\,[\kk-\overline{\kk}]\,.
    	\end{equation}
    	
    	Finally, by examining the integral expressions $J[\phi_{{}^{{}^{{}_{{}^{(i)}}}}}\hskip-0.06cm]$ and using the fact that the dominant radial dependence of $\dd$ implies ${\underline{\sqrt{\dd}}}_{\,[ {\tiny{\interior{\phi}_{{}^{{}^{{}_{{}^{(i)}}}}}\hskip-0.06cm}}]}\sim r^2$, we conclude that the quasi-local angular momentum expressions $J[\phi_{{}^{{}^{{}_{{}^{(i)}}}}}\hskip-0.06cm]$, and thus also their asymptotic limit and thus also the Bondi angular momentum, cannot be finite unless for the coefficients of the logarithmic terms in the asymptotic expansion of $\kk$ \eqref{eq: lemmaBJ} holds, which completes the proof of this lemma.
    	\end{proof}
    	
    	\medskip
    	
    	Returning to the main stream of the proof of our theorem, recall that some of the logarithmic terms in \eqref{eq:Nhrseries-0}-\eqref{eq:krseries-0} vanish (see \eqref{eq: lemmaBM5} and \eqref{eq: lemmaBJ}) and for some others additional algebraic relations hold (see \eqref{eq: lemmaBM6} and \eqref{eq: lemmaBM7}) when the Bondi mass and the angular momentum are guaranteed to be finite. However, the non-vanishing and undetermined coefficients in \eqref{eq:Nhrseries-0}-\eqref{eq:krseries-0} are still completely arbitrary, more importantly they are still not subject to the parabolic-hyperbolic system \eqref{eq:phN}-\eqref{eq:phK}, and this is indeed the case we are interested in. To obtain the desired restrictions, we substitute the asymptotic expansions in \eqref{eq:Nhrseries-0}-\eqref{eq:krseries-0} into the parabolic-hyperbolic system \eqref{eq:phN}-\eqref{eq:phK} and sort the terms with respect to powers of $r^{-1}$ and also of $\log\,r^{-1}$. Since the resulting leading order coefficients of the powers of $r^{-1}$ and $\log\,r^{-1}$ are all independent of $r$, the parabolic-hyperbolic system \eqref{eq:phN}-\eqref{eq:phK} holds up to the desired finite orders if these coefficients vanish individually. In this way we obtain a system of algebraic equations for some of the coefficients in the asymptotic expansions \eqref{eq:Nhrseries-0}-\eqref{eq:krseries-0}.
    	
    	The number of terms still involved in \eqref{eq:falloffT12NNom}-\eqref{eq:falloffT12Kcqqom} is optimal to produce the desired system of algebraic equations capable of determining the coefficients involved in \eqref{eq:Nhrseries-0}-\eqref{eq:krseries-0}. Some of these (compact enough) coefficients, obtained from the algebraic system discussed above, are
    	\begin{align}
    		&\hskip-0.2cm\Nhat_0=0, \ \Nhat_1=\kkappa_0^{-1}, \hskip0cm
    		\ \Nhat_2=-\frac{\KK_1+\aaa{}_{^{(-1)}}\kkappa_0}{2\kkappa_0^2}\,, \ \Nhat_3 = -\tfrac32\frac{(\eth\kkappa_0)(\ethb\kkappa_0)}{\kkappa_0^5} + \tfrac12\frac{\eth\ethb\kkappa_0}{\kkappa_0^4}\nonumber\\& \hskip0.9cm + \tfrac18\frac{2 \KK_1^2 + {\Kcqq}{}_{^{(-1)}}\overline{\Kcqq}{}_{^{(-1)}}-4}{\kkappa_0^3} -\tfrac18\frac{8  \aaa_0+2(\eth\overline{\NN_1}+\ethb\NN_1)+4\aaa{}_{^{(-1)}}^2+5 \bb{}_{^{(-1)}}\overline{\bb{}_{^{(-1)}}}}{\kkappa_0} \nonumber\\& \hskip0.9cm +\tfrac14\frac{4 \kkappa_2+2(\NNbar_1\eth\kkappa_0+\NN_1\ethb\kkappa_0)+2\KK_1\,\aaa{}_{^{(-1)}}+(\Kcqq{}_{^{(-1)}}\overline{\bb}{}_{^{(-1)}}+\overline{\Kcqq}{}_{^{(-1)}}\bb{}_{^{(-1)}})}{\kkappa_0^2}\,, \label{eq:rasympNh}\\
    		&\hskip-0.2cm\KK_0=2\kkappa_0, \ \KK_1=\KK_1,\ \KK_2=\frac{\eth\kkappa_0\ethb\kkappa_0}{\kkappa_0^3}+\ethb\eth\kkappa_0^{-1}-(\NNbar_1\eth\kkappa_0+\NN_1\ethb\kkappa_0)-2\kkappa_2 \nonumber\\&\hskip5.3cm-\tfrac12\KK_1\aaa{}_{^{(-1)}}-\tfrac14(\Kcqq{}_{^{(-1)}}\overline{\bb}{}_{^{(-1)}}+\overline{\Kcqq}{}_{^{(-1)}}\bb{}_{^{(-1)}})\,,\label{eq:rasympKK}\\
    		&\hskip-0.2cm\kk_0=\frac{\eth\kkappa_0}{\kkappa_0},\ \kk_1=\frac{[\kkappa_0\,\eth\KK_1-\KK_1\,\eth\kkappa_0]-[\kkappa_0\,\overline{\eth}\Kcqq{}_{^{(-1)}}-\Kcqq{}_{^{(-1)}}\overline{\eth}\kkappa_0]}{2\kkappa_0^{2}}, \ \kk_2=\kk_2,\label{eq:rasympkk}
    	\end{align}
    	Remarkably, on the way to determine these relations, and to close the algebraic system, we also get for the coefficients of the logarithmic terms in the asymptotic expansions \eqref{eq:Nhrseries-0}-\eqref{eq:krseries-0}
    	\begin{align}
    		\Nhat_{3,j}^{[log]} & = \Nhat_{4,j}^{[log]}=\KK_{2,j}^{[log]}=\KK_{3,j}^{[log]}=\KK_{4,j}^{[log]}=\kk_{3,j}^{[log]}=0\,, \label{eq: lemmaPH} 
    	\end{align}
    	and, in addition $\kkappa_{1}=0$ verifying the third relation in \eqref{eq:alg-theorem}.
    	
    	Note that due to the nonlinearity involved, only a finite number of terms of the asymptotic expansions \eqref{eq:Nhrseries-0}-\eqref{eq:krseries-0} can be derived using the algorithm outlined above. However, the coefficients in \eqref{eq:rasympNh}-\eqref{eq:rasympkk} are sufficient to obtain a considerable simplification of the expansion coefficients in \eqref{eq: exp-thetaPM}. As expected, $\Nhat_1$ and $\Nhat_2$ in the equations \eqref{eq: lemmaBM1} and \eqref{eq:rasympNh}, respectively, by virtue of $\KK_0=2\,\kkappa_0$, are consistently equal to each other, so $f_{-2}=f_{-1}=0$ automatically holds. Surprisingly, however, the two forms of $\Nhat_3$, given by the equations \eqref{eq: lemmaBM1} and \eqref{eq:rasympNh}, coincide only if the relation \begin{equation}
    		(\Kcqq{}_{^{(-1)}}+\bb{}_{^{(-1)}}\,\kkappa_0)\overline{(\Kcqq{}_{^{(-1)}}+\bb{}_{^{(-1)}}\, \kkappa_0)}=0\,,
    	\end{equation}
    	which implies that $f_{0}=-4$ can only be satisfied, so the Bondi mass can only be finite if
    	\begin{equation}\label{eq:relKCqq}
    		\Kcqq{}_{^{(-1)}}+\bb{}_{^{(-1)}}\,\kkappa_0=0\,,
    	\end{equation}
    	holds.
        
        Note that the Bondi mass is then determined by the first-order contribution of the product of $\sqrt{\mathcal{A}/16\,\pi}=r/2+\OO{0}$ and the integral of $f_1$ in \eqref{eq: exp-thetaPM}. After assuming the vanishing of $\aaa{}_{^{(-1)}}$, $\bb{}_{^{(-1)}}$ and $\Kcqq{}_{^{(-1)}}$ for simplicity, we get that
    	\begin{align}
    		\hskip-0.cm  m_B & =\frac{1}{32\pi}\hskip-0.1cm\int_{\partial\widetilde{\Sigma}}\bigg[\frac{\KK_1{}^3}{\kkappa_0}+8\,\Nhat_4\,\kkappa_0{}^3
    		-4\,\kkappa_0\,\kkappa_3+12\,\aaa_1\,\kkappa_0{}^2-2\,\KK_1\,\big[\kkappa_0{}^{-1}
    		+3\,\aaa_0\,\kkappa_0 - 4 \,\kkappa_2\big]
    		\bigg.\nonumber\\
    		&\bigg.\hskip0.6cm
    		-\Big\{\kkappa_0{}^5\,\NN_1{}^{-1/2}\,\ethb\,\left(\NN_1{}^{3/2}\kkappa_0{}^{-4}\KK_1\right)+2\,\kkappa_0\,\NN_2{}^2\,\ethb\left(\NN_2{}^{-1}\kkappa_0\right)+\tfrac14\,\kkappa_0{}^{5}\,\eth\ethb\left(\kkappa_0{}^{-6}\,\KK_1\right)
    		\Big.\bigg.\nonumber\\
    		&\bigg.\Big.\hskip5.2cm
    		 -\tfrac32\,\kkappa_0{}^{-2}\,\KK_1\,\eth\ethb\,\kkappa_0+\tfrac14\,\kkappa_0\,\eth\ethb\,\KK_1 +``cc"
    		\Big\}
    		\bigg]\boldsymbol{\epsilon}_q\,.
    	\end{align}
    	This expression depends only on the background fields, and also on $\KK_1$ and $\Nhat_4$---which are part of asymptotic freedom, as we will see in the proof of Theorem \ref{main theorem 2}---, thereby it is manifestly finite.
    	
    	\medskip
    	
    	An analogous explicit formula is also needed for the Bondi angular momentum expression. Before giving it, recall that the quasi-local angular momentum on the topological two-sphere, $\mathscr{S}_r$, starts by choosing a centre-of-mass unit sphere reference system on $\mathscr{S}_r$. One must also find the axial vector field $\phi{}^a$ so that the functional
    	\begin{equation}
    		J[\phi] = - (8\, \pi)^{-1}\int_{{\mathscr{S}}_r} (\phi{}^a{ \rm \bf k}_a) \, \widehat{\bold{\epsilon}} \,,
    	\end{equation}
    	attains its maximum on $\mathscr{S}_r$ for a given choice of $\kk_a$ \cite{Racz:2024giv}.
    	
    	There are two important facts to note here. First, for an axial vector field the divergence $\widehat{D}_a \phi{}^a$ vanishes. This has the important consequence that if $\kk_a=\kk'_a+\widehat{D}_a \chi$, then the gradient $\widehat{D}_a \chi$ of the scalar function $\chi$ does not contribute to the value of $J[\phi]$ (see, e.g., the derivation of equation (1.3) in \cite{Racz:2024giv}). Second, even after dropping all the logarithmic terms $\kk_{1,j}^{[log]}$, $\kk_{2,j}^{[log]}$ and $\kk_{3,j}^{[log]}$ that appear in \eqref{eq:krseries-0}, because of the relation $\widehat{\boldsymbol{\epsilon}}\sim r^2\,{\boldsymbol{\epsilon}}_q$, we have to guarantee that the first two terms, $\kk_0$ and $\kk_1$, in the asymptotic expansion of $\kk$ either vanish or they are entirely or partialy gradients. Note that  \eqref{eq:rasympkk} gives that $\kk_0=\eth(\log\kkappa_0)= q^e\widehat{D}_e (\log\kkappa_0)$ and $\kk_1=\eth\big[\tfrac12\,\kkappa_0^{-1}\KK_1\big]-\overline{\eth}\big[\tfrac12\,\kkappa_0^{-1}\Kcqq{}_{^{(-1)}}\big]=q^e\widehat{D}_e[\tfrac12\,\kkappa_0^{-1}\KK_1\big]-\overline{\eth}\big[\tfrac12\,\kkappa_0^{-1}\Kcqq{}_{^{(-1)}}\big]$. Note also that ${\kk}_a = \tfrac12\, [q_a \bar{q}^e+\bar{q}_a q^e]\,{\kk}_e=\tfrac12\, [\bar{q}_a\,{\kk} + q_a\,\bar{\kk}]$, which implies for $\widetilde{\kk}=\eth\chi=q^e\widehat{D}_e\chi$ that $\widetilde{\kk}_a = \tfrac12\, [q_a \bar{q}^e+\bar{q}_aq^e]\,\widehat{D}_e\chi=\gammahat_a{^e} \,\widehat{D}_e\chi=\,\widehat{D}_a\chi$, where the equivalence of the projectors $q_a{^b}=\tfrac12\, [q_a \bar{q}^b+\bar{q}_a q^b]$ and  $\gammahat_a{^b}$ was also used. These observations imply that $\kk_0$ and, whenever $\Kcqq{}_{^{(-1)}}=0$\,\footnote{To see this, note that $\overline{\eth}\big[\kkappa_0^{-1}\Kcqq{}_{^{(-1)}}\big]=0$ implies the vanishing of the product $\kkappa_0^{-1}\Kcqq{}_{^{(-1)}}$, which, due to the non-vanishing of $\kkappa_0$, gives $\Kcqq{}_{^{(-1)}}=0$. Note also that in principle, by requiring the vanishing of $\kk_1$, we could also end up with the assumption that $\eth\big[\kkappa_0^{-1}\KK_1\big]-\overline{\eth}\big[\kkappa_0^{-1}\Kcqq{}_{^{(-1)}}\big]=0$. However, this would result in a restriction of the constrained field $\KK$, which would lead to a conceptual inconsistency. Therefore, we must discard this case.}, also the remaining part of $\kk_1$ can be given as gradients of scalar fields, respectively, which, in turn, imply that the truly finite Bondi angular momentum read as
    	\begin{equation}\label{eq: Bondi-impmom}
    		J_B[\phi] = - (16\, \pi)^{-1}\int_{\partial\widetilde{\Sigma}} [(\phi^a\overline{q}_a)\,\kk_2 + (\phi^aq_a)\,\overline{\kk}_2] \, {\boldsymbol{\epsilon}}_q \,.
    	\end{equation}
    	Note that, by virtue of \eqref{eq: Bondi-impmom}, the Bondi angular momentum is determined by $\kk_2$, which, like $\KK_1$ and $\Nhat_4$ in the  Bondi mass case, is part of the asymptotic freedom, as we will see in the proof of Theorem \ref{main theorem 2}.
    	Note also that \eqref{eq:relKCqq} together with the vanishing of $\Kcqq{}_{^{(-1)}}$ implies $\bb{}_{^{(-1)}}=0$, i.e., the first and middle relations in \eqref{eq:alg-theorem} also hold.
    		
    	\medskip
    	
    The other direction is self-explanatory, as if all coefficients of the logarithmic terms in \eqref{eq:Nhrseries-0}-\eqref{eq:krseries-0} vanish up to order four and three for $\Nhat,\KK$ and $\kk$, respectively, and, in addition, the relations in \eqref{eq:alg-theorem} hold for smooth solutions of the parabolic-hyperbolic form of the constraint equations that satisfy the assumptions in our theorem, the Bondi mass and angular momentum are automatically well-defined, which completes our proof.
	\end{proof} 
	
	Since in Lemma \ref{lemma: Bondi-mass} it was found that each of the logarithmic terms $\Nhat_{1,j}^{[log]}$ vanishes, by \eqref{eq:lKrr} and \eqref{eq:lKra} of Lemma \ref{lemma1} we  have immediately that the relations $\tildeon{K}{}_{rr}^{[log]}=0$, $\tildeon{K}{}_{ra}^{[log]}=0$ in \eqref{eq: conds-Prop3} hold on $\partial\widetilde{\Sigma}$ if the Bondi mass is  well-defined.\footnote{Note that this in principle would still allow for the presence of log terms, which is in complete accordance with the findings in \cite{Chrusciel:1998he,Hintz:2017xxu}, where the authors  was shown that the Trautman-Bondi energy remains well-defined for a wide class of "polyhomogeneous" metrics. As also proved in \cite{Hintz:2017xxu}, this is true even for the Bondi mass loss formula.} If, in addition, the considered initial data is subject to the parabolic-hyperbolic form of the constraints the remaining third relation ${\tildeon{\KK}}{}{}^\circ_{ab} - \tildeon{\widehat{K}}{}^\circ_{ab}=0$ of Lemma \ref{lemma1}, which by virtue of $\Nhat_1=\kkappa_0{}^{-1}$, \eqref{eq:shear0}, \eqref{eq:relKCqq} also holds on $\partial\widetilde{\Sigma}$.

	Summarizing these observations, together with those made in proving Proposition \ref{prop3}, Lemma \ref{lemma1}, and Theorem \ref{main theorem 1}, we also have the following:
	\begin{corollary}\label{corollary}
		Consider an initial data set $(\Nhat,\KK,\kk; \NN,\aaa,\bb, \kkappa,\Kcqq)$ that satisfies all the conditions applied above in Theorem \ref{main theorem 1}. Suppose also that this initial data set admits a well-defined Bondi mass.
		Then all three relations in \eqref{eq: conds-Prop3} hold on $\partial\widetilde{\Sigma}$, and, in turn, the Cauchy development of the considered asymptotically hyperboloidal initial data admits a smooth conformal boundary.
	\end{corollary}
	
	\subsection{Smooth hyperboloidal initial data}
	\label{sec:main-smooth}
	
	Before presenting our second main result, some clarification of the applied terminology used is in order, which will greatly simplify our discussion here and later. Recall first that we are interested in the asymptotic behavior of the initial data on $\Sigma=(r_0,\infty)\times \mathbb{S}^2$, with $r_0>0$, covered by the local coordinates $(r,\vartheta,\varphi)$. It is worth relabeling the $r=const$ level surfaces with the $t=\omega=1/r$ relation, and also, by slightly abusing the notation, replacing $\omega\circ\Phi$ with `$t$' on $\Sigma$, which replaces the factor $(r_0,\infty)$ in the topological product $\Sigma=(r_0,\infty)\times \mathbb{S}^2$ with $(0,t_0)$, where $t_0=1/r_0$. We can then use as synonyms for $\widetilde{\Sigma}$ and $\partial\widetilde{\Sigma}$ the topological products $[0,t_0)\times \mathbb{S}^2$ and $\{0\}\times \mathbb{S}^2$, respectively. This provides a shorthand notation for the embedding, $\Phi: \Sigma\rightarrow \widetilde{\Sigma}\setminus\partial\widetilde{\Sigma}$, by applying $(r,\vartheta,\varphi)\mapsto (t,\vartheta,\varphi)$. Since the critical direction is the radial one we will henceforth assume that in the tangential directions all the variables are smooth on the $t=const$ level surfaces, which are topological two-spheres, as indicated by using the symbol $C^\infty(\mathbb{S}^2)$. More importantly, $C^k\bigl({(0,t_0)},C^\infty(\mathbb{S}^2)\bigr)$ and $C^k\bigl({[0,t_0)},C^\infty(\mathbb{S}^2)\bigr)$, for some non-negative integer $k$, denotes that a field under consideration is of class $C^k$ in the radial direction on ${\Sigma}$ and $\widetilde{\Sigma}$, respectively. $C^\infty$ replaces $C^k$ whenever the fields are also smooth in the radial direction.
	
	We will now combine the falloff conditions on the free data $(\NN,\aaa,\bb; \kkappa,\Kcqq)$ and on the constrained variables $(\Nhat,\KK,\kk)$ found in Theorem \ref{main theorem 1}, which guarantee the existence of well-defined finite Bondi mass and angular momentum, to spell out the conditions on the initial data used in our second main result. As we will see, these conditions are sufficient to ensure that the solutions of the parabolic-hyperbolic form of the constraint equations are free of logarithmic singularities. Note that the authors imposed rather strong falloff conditions on the free data in Proposition 3 of \cite{Beyer:2021kmi}, and  did so without any real attempt to explain the reason for using these restrictions. Note also that in the proof of our second main result, despite the fact that our setup is more general by allowing the involvement of $\aaa{}_{^{(-1)}}$, $\bb{}_{^{(-1)}}$ and $\Kcqq{}_{^{(-1)}}$, and also the angular dependence of $\kkappa_0$, we will still be able to save a considerable part of the Fuchsian analyses of Proposition 3 of \cite{Beyer:2021kmi} of Beyer and Ritchie, which will be done below, pointing out the differences.
    
    \begin{theorem}\label{main theorem 2}
    Choose a generic asymptotically hyperboloidal set of free data $(\NN,\aaa,\bb, \kkappa,\Kcqq)$ on $\Sigma$ (not necessarily derived from a solution of the vacuum constraints) which satisfies the falloff conditions as in Theorem \ref{main theorem 1} with $\kkappa_0$ being a strictly positive smooth on $\partial\widetilde{\Sigma}$. 
    Suppose that on $\Sigma$ $(\Nhat,\KK,\kk)$ are of class  $C^\infty\bigl({(0,t_0)},C^\infty(\mathbb{S}^2)\bigr)$ solutions of the parabolic-hyperbolic form of the constraints \eqref{eq:phN}-\eqref{eq:phK}, whose coefficients are derived from the chosen free data such that $\Nhat>0$. Then, the constrained fields $(\Nhat,\KK,\kk)$ are all of class  $C^\infty\bigl({[0,t_0)},C^\infty(\mathbb{S}^2)\bigr)$ on the whole of $\widetilde{\Sigma}$, i.e., no logarithmic singularities occur, if and only if the asymptotically hyperboloidal initial data set under consideration admits well-defined Bondi mass and angular momentum, and, in addition, \eqref{eq:alg-theorem} and also the following two relations
    	\begin{align}
     \eth\overline{\eth} \,\aaa{}_{^{(-1)}}&=0\,,\label{eq:extra2}\\
    	\overline{\eth}\left[{\Kcqq}{}_0\cdot\kkappa_0^{-1}-\tfrac{1}{2}\,\eth\eth\,\kkappa_0^{-2}\right]&=0\,,\label{eq:extra1}
    	\end{align}
    	hold on $\partial\widetilde{\Sigma}$. 
    \end{theorem}

    A few comments are in order.
    First, note that throughout \cite{Beyer:2021kmi}, especially in the statement of Proposition 3 of \cite{Beyer:2021kmi}, it was assumed that $\aaa{}_{^{(-1)}}$, $\bb{}_{^{(-1)}}$, and $\Kcqq{}_{^{(-1)}}$ all vanish, and that $\kkappa_0$ is constant over $\partial\widetilde{\Sigma}$. This then implies the vanishing of $\eth\kkappa_0$, and in particular that \eqref{eq:extra1} reduces to $\ethb{\Kcqq}{}_0=0$. Note also that neither the vanishing of the expansion coefficient $\aaa{}_{^{(-1)}}$ nor the constancy of $\kkappa_0$ is required in our result, thus, Theorem \ref{main theorem 2} provides a generalization of Proposition 3 of \cite{Beyer:2021kmi}.
    
    Second, it should also be clear that the \eqref{eq:alg-theorem}, \eqref{eq:extra2}, and \eqref{eq:extra1}, are conditions on the free data only. Note also that \eqref{eq:extra2} implies that
    \begin{equation}
    	\aaa{}_{^{(-1)}}=const \label{eq:extra2b}\\
    \end{equation}		
    throughout $\partial\widetilde{\Sigma}$, while \eqref{eq:extra1} can also be solved for ${\Kcqq}{}_0$ and we get
    \begin{equation}
    	{\Kcqq}{}_0 =\tfrac{1}{2}\,\kkappa_0\cdot\eth\eth\,\kkappa_0^{-2}\,.\label{eq:extra1b}
    \end{equation}
  
  \begin{proof}\hskip-0.35cm({\bf of Theorem \ref{main theorem 2}})
  It follows from Theorem \ref{main theorem 1} that the constrained variables $\Nhat,\KK$ and $\kk$, which are supposed to be smooth on $\Sigma$ take the form of the asymptotic expansions
  \begin{align}
  	\Nhat(t)&=\Nhat_0+\Nhat_1t^{1}+\Nhat_2t^{2}+\Nhat_3t^{3}+\Nhat_4t^{4}+w_{\Nhat}(t)t^{4}\,,\label{eq:Nhseries}\\
  	\KK(t)&=\KK_0+\KK_1t^{1}+\KK_2t^{2}+\KK_3t^{3}+\KK_4t^{4}+w_\KK(t)t^{4}\,,\label{eq:Kseries}\\
  	\kk(t)&=\kk_0+\kk_1t^{1}+\kk_2t^{2}+\kk_3t^{3}+w_\kk(t)t^{3}\,,\label{eq:kseries}
  \end{align}
  where the expansion coefficients are of class $C^\infty(\mathbb{S}^2)$ on $\partial\widetilde{\Sigma}\sim \mathbb{S}^2$, while the residuals, $w_{\Nhat}(t), w_\KK(t),w_\kk(t)$, are at least of class $C^0\bigl({[0,t_0)}, C^\infty(\mathbb{S}^2)\bigr)$ and they all vanish at $\partial\widetilde{\Sigma}$. Note that the asymptotic expansions  in  \eqref{eq:Nhseries}-\eqref{eq:kseries} are compatible with our results in Theorem \ref{main theorem 1}
  \begin{align}
  	\Nhat_{1,j}^{[log]} & = \Nhat_{2,j}^{[log]} = \Nhat_{3,j}^{[log]} = \Nhat_{4,j}^{[log]} =0\,, \label{eq:logNh}\\
  	\KK_{1,j}^{[log]} & =\KK_{2,j}^{[log]}=\KK_{3,j}^{[log]}=\KK_{4,j}^{[log]} =0\,,  \label{eq:logKK}\\
  	\kk_{1,j}^{[log]} & =\kk_{2,j}^{[log]}=\kk_{3,j}^{[log]}=0\,, \label{eq:logkk} 
  \end{align}
  and they still allow the potential involvements of higher order polyhomogeneous terms that can be covered by the residuals terms $w_{\Nhat}(t,p), w_\KK(t,p), w_\kk(t,p)$.  

  Substitute \eqref{eq:Nhseries}-\eqref{eq:kseries}, along with the asymptotic expansions \eqref{eq:falloffT12NNom}-\eqref{eq:falloffT12Kcqqom},  into the parabolic-hyperbolic system \eqref{eq:phN}-\eqref{eq:phK} and sort the terms now with respect to powers of $t$. Since the resulting leading order coefficients of the powers of $t$ are all independent of $t$, the parabolic-hyperbolic system \eqref{eq:phN}-\eqref{eq:phK} holds up to the chosen finite orders if these coefficients vanish individually. In this way we obtain a system of algebraic equations for the primary expansion coefficients used in \eqref{eq:Nhseries}-\eqref{eq:kseries}.

  The above algorithm eventually breaks down, because only a finite number of terms are involved in the asymptotic expansions \eqref{eq:Nhseries}-\eqref{eq:kseries}. This is compensated by the inclusion of the $t$-dependent residual terms $w_{\Nhat}(t), w_\KK(t),w_\kk(t)$.
  
  The number of terms in \eqref{eq:falloffT12NNom}-\eqref{eq:falloffT12Kcqqom} is optimal to produce the desired system of algebraic equations capable of determining the coefficients involved in \eqref{eq:Nhseries}-\eqref{eq:kseries}. As it is expected, in this way we recover the relations in \eqref{eq:rasympNh}-\eqref{eq:rasympkk}.
  
  Note that besides the coefficients in the asymptotic expansions \eqref{eq:Nhseries}-\eqref{eq:kseries} we also need the extra relation  \eqref{eq:extra1} (see also \eqref{eq:extra1b}) to close the system. Inspection of these expressions shows that some of the coefficients in the asymptotic expansions \eqref{eq:Nhseries}-\eqref{eq:kseries} are completely free. Consistent with the results in \cite{Beyer:2021kmi}, there are no restrictions on $\KK_1$, $\kk_2$, and $\Nhat_4$, which thus represent the asymptotic freedom of the initial data set. All other coefficients in the asymptotic expansions \eqref{eq:Nhseries}-\eqref{eq:kseries} are completely determined by the asymptotic freedom $(\KK_1,\kk_2,\Nhat_4)$ and free data $(\NN,\aaa,\bb;\kkappa,\Kcqq)$, together with terms derived from them.

  To prove that the constrained variables $(\Nhat,\KK,\kk)$  do indeed extend smoothly to $\partial\widetilde{\Sigma}$, we can repeat the Fuchsian analysis-based argument of Beyer and Richtie \cite{Beyer:2021kmi}. First, note that by eliminating the expansion coefficients from the parabolic-hyperbolic equations, and imposing \eqref{eq:alg-theorem} and \eqref{eq:extra2} (see also \eqref{eq:extra2b}), we end up with a Fuchsian partial differential equation (see, e.g., \cite{Beyer:2019cwu})  of the form
  \begin{align}\label{eq: Fuchsian1}
  	\partial_t \underline{W}(t,p) = {} & \frac{1}{t} diag\big[-3,-1,0\big]\times \underline{W}(t,p) \nonumber \\ & + \underline{H}\big(t,p; \KK_1(p),\kk_2(p),\Nhat_4(p),\underline{W}(t,p),\eth\underline{W}(t,p),{\bar\eth}\underline{W}(t,p),\eth{\bar\eth}\underline{W}(t,p)
  	\big)    
  \end{align}
  for the vector-valued variable $\underline{W}(t,p)=\big( w_\KK(t,p),w_\kk(t,p),w_{\Nhat}(t,p)\big)$, for all $t\in(0,t_0)$ and $p\in \mathbb{S}^2$, where $\underline{H}$ is a complicated but explicitly known vector-valued function that is smooth in each of its specified arguments and it regularly extends to $\partial\widetilde{\Sigma}$. The solution of \eqref{eq: Fuchsian1} can then be given as
  \begin{equation}\label{eq: Fuchsian2}
  	\underline{W}(t,p) = diag\big[t^{-3},t^{-1},1\big] \times \int_0^t diag\big[s^3,s,1\big]\times \underline{H}(s,p) \,\mathrm{d}s\,.
  \end{equation}
  Since the integrand regularly extends to $s=0$, we can also perform the integral transformation by replacing $s$ with the product $t\cdot\tau$, which by \eqref{eq: Fuchsian2} yields 
  \begin{equation}\label{eq: Fuchsian3}
  	\frac{1}{t}\,\underline{W}(t,p) = \int_0^1 diag\big[\tau^3,\tau,1\big]\times \underline{H}(t\cdot\tau,p) \,\mathrm{d}\tau\,.
  \end{equation}
  Since the integrand on the right hand side of \eqref{eq: Fuchsian3} is also regular over $\widetilde{\Sigma}$, this implies that not only the right hand side but the left hand side of \eqref{eq: Fuchsian3} must also be regular. This then implies that both terms on the right hand side of \eqref{eq: Fuchsian1} are regular and in turn the first order $t$-derivative $\partial_t\underline{W}$ of the vector-valued variable of the residuals $\underline{W}(t,p)=\big( w_\KK(t,p),w_\kk(t,p),w_{\Nhat}(t,p)\big)$ is regular at $t=0$. By repeating this process inductively we can also prove that $t$-derivatives of the vector-valued variable $\underline{W}(t,p)$ up to arbitrary order extend regularly to $\partial\widetilde{\Sigma}$, which in turn implies that the constrained variables $(\Nhat,\KK,\kk)$ extend smoothly to $\partial\widetilde{\Sigma}$, i.e., they are of class $C^\infty\bigl({[0,t_0)},C^\infty(\mathbb{S}^2)\bigr)$ over $\widetilde{\Sigma}$.
  
  \medskip
  
  The other direction is again self-explanatory: if the constrained fields $\Sigma$ $(\Nhat,\KK,\kk)$ are all of class $C^\infty\bigl({[0,t_0)},C^\infty(\mathbb{S}^2)\bigr)$ on the whole $\widetilde{\Sigma}$, then the asymptotically hyperboloidal initial data under consideration admit well-defined Bondi mass and angular momentum, and in addition \eqref{eq:alg-theorem}, \eqref{eq:extra2} and \eqref{eq:extra1} all hold on $\partial\widetilde{\Sigma}$.
  \end{proof}

  \section{Numerical studies}
  \label{sec:numerical}

  The purpose of this section is to present our numerical results, the main role of which is to provide additional evidence in support of our claims in the previous section. First, in Section \ref{sec:background} we present our choice of free data. Then, in Section \ref{sec:numsetup} we give a brief overview of the numerical methods used and the setting of the applied parameters. In Section \ref{sec:asymrel} we verify that various combinations of the constrained fields fall with the rates derived based on their anticipated analytic behavior. In Section \ref{subsec: approx Hawking} we present some results concerning the numerical evaluation of the Hawking and Bondi masses.

  \subsection{The Kerr background}
  \label{sec:background}

  Among the $12$ scalar variables stored in the initial data $(\Nhat,\NN,\aaa,\bb;\kkappa,\kk,\KK,\Kcqq)$ only $4$ are constrained, we have to specify the remaining $8$ as free data. To determine this free data we start with the Kerr metric in $(t,r,\vartheta,\varphi)$ spherical Kerr coordinates (see \cite{Csukas:2023aya}), with mass parameter $M$, and rotation parameter $a$. We construct $T=const$ hyperboloidal time slices of this metric using the relation $t=T+\sqrt{M^2+r^2}$. \footnote{This slice was inspired by leaving out the $\log$ term from the RT coordinate transformation in \cite{Racz:2024jws}, i.e.
  	$$\tau=T+M\,\frac{M^2+R^2}{M^2-R^2}\left(-4M\log(|1-R^2/M^2|)\right),$$
  	$$r=\frac{2R}{1-R^2/M^2}.$$}

  We proceed with a $3+1$ decomposition with respect to $T=const$ surfaces, then a $2+1$ decomposition with respect to $r=const$ surfaces. The explicit expressions for the variables can be found in the accompanying Mathematica notebook \cite{csukas_2025_14779011}.
  The asymptotic behavior of the corresponding background\,\footnote{These are indicated by the front upper zero index in round brackets,  ``${}^{\scriptscriptstyle(\!0\!)}$''.} 
  spin-weighted variables are
  \begin{align}
  	\bgNh&=\frac{\sqrt{M^2-a^2\sin^2\vartheta}}{r}-\frac{M^4-a^4\sin^2\vartheta}{2r^3\sqrt{M^2-a^2\sin^2\vartheta}}+\OO{-4}\label{eqs: bgNh}\\
  	\bgNN&=-\mathbbm{i} a M\sin\vartheta\frac{M^2-2a^2\sin^2\vartheta}{r^3}+\OO{-4} \label{eqs: bgNN}\\
  	\bgaaa&=r^2+\frac{1}{2}a^2(1+\cos^2\vartheta)+\OO{-1}\label{eqs: bgaa}\\
  	\bgbb&=-\frac{1}{2}a^2\sin^2\vartheta-\frac{a^2M}{r}\sin^2\vartheta+\OO{-2} \label{eqs: bgbb}\\
  	\bgkkappa&=\frac{1}{\sqrt{M^2-a^2\sin^2\vartheta}}+\frac{a^4\sin^2(2\vartheta)}{8r^2(M^2-a^2\sin^2\vartheta)^{3/2}}+\OO{-3} \label{eqs: bgkappa}\\
  	\bgkk&=-\frac{a^2\cos\vartheta\sin\vartheta}{M^2-a^2\sin^2\vartheta}+\OO{-2}\\
  	\bgKK&=\frac{2}{\sqrt{M^2-a^2\sin^2\vartheta}}+\OO{-2} \label{eqs: KKbg}\\
  	\bgKcqq&=\frac{a^2\sin^2\vartheta}{\sqrt{M^2-a^2\sin^2\vartheta}}+\frac{3a^2M\sin^2\vartheta}{r\sqrt{M^2-a^2\sin^2\vartheta}}+\OO{-2}\,.\label{eqs: Kcqqbg}
  \end{align}
  It is easy to see, that these expressions satisfy all the requirements posed in \eqref{eq:falloffT12NNom}-\eqref{eq:falloffT12Kcqqom}. It is also straightforward to verify that \eqref{eq:extra1} holds for this background.

  In the remainder of this paper, we aim to provide numerical evidence to support the expectation that the perturbed Kerr initial data constructed using part of this background data as freely specifiable variables does indeed lead to a hyperboloidal initial data that smoothly extends to $\partial\widetilde{\Sigma}$.
  
  \subsection{Applied numerical setup}
  \label{sec:numsetup}

  In constructing the perturbed Kerr initial data we utilize the code \verb|ConstraintSolver| \cite{Csukas:2023aya,noauthor_csukas_2024}. All the data and code used to produce the figures---including an archive of the specific version of \verb|ConstraintSolver|---are available to the public \cite{csukas_2025_14779011}.

  Our code uses a spin-weighted spherical harmonic expansion with cutoff $\ell_{max}=16$ in the angular domain, and a fourth-order accurate adaptive Runge-Kutta-Fehlberg method in the radial direction with tolerance parameter $\epsilon=10^{-9}$.
  
  The integration starts at $r=1$, 
  and ends at $r=10^5$, spanning five orders of magnitude, which gives a good indication of the asymptotic behavior. Note that starting at $r=1$ ensures that the integration along the $T=const$ slice starts in the black hole region. We set the initial data for $\KK$ at $r=1$ by adding $10\cdot Y_2{}^0$ excitation to the background $\bgKK$ indicated by \eqref{eqs: KKbg}. We will see that this excitation of the Kerr background produces a Bondi mass that is about 1.87 of the ADM mass $M$ (which is also the Bondi mass) of the stationary Kerr background.

  Finally, in order to increase the numerical accuracy, we divide the fields $(\KK,\kk,\Nhat)$ into a known background $(\bgKK,\bgkk,\bgNh)$ and deviations from the background $(\dKK,\dkk,\dNh)$ and solve the deviations directly, as it was done earlier in \cite{Nakonieczna:2017eev,Csukas:2019qco,Csukas:2023aya}.

  \subsection{Numerical check for the falloff of various expressions}
  \label{sec:asymrel}

  Note that the analytical part of our second main result provides valuable insights into the interrelations of the constrained fields. More explicitly, the asymptotic behavior of the fields $(\Nhat, \KK,\kk)$, together with the dependence of the expansion coefficients on the asymptotic degrees of freedom, determines the falloff rates for various combinations of the constraint fields, which can be used as an accuracy check of our numerical scheme. 

  The most obvious of these relations follows from the fact that the perturbation does not contribute to the leading order. For example, for $\KK$, the leading order coefficient, $\KK_0=2\kkappa_0$, is completely determined by the free data. The coefficient next to the leading order, $\KK_1$, is the first asymptotic freedom where the perturbations appear. In fact, Fig.\,\ref{fig:falloff:KK} clearly shows that $\dKK$ falls off as $r^{-1}$. Now, considering $\kk$, the leading order coefficient, $\kk_0=\kkappa_0^{-1}\eth\kkappa_0$, is again determined by the free data, and the perturbation appears next to the leading order by the asymptotic freedom represented by $\KK_1$. Fig.\,\ref{fig:falloff:Rekk} shows that this is indeed the case. Since $\KK$ is real, $\KK_1$ is also real, as is $\kkappa_0$. Since in the case of axial symmetry the action of the "eth" operator, $\eth$, is only a derivative with respect to the polar angle $\vartheta$, we have that both $\eth\KK_1$ and $\eth\kkappa_0$ are real. This means that $\kk_1=\eth\big[\tfrac12\,\kkappa_0^{-1}\KK_1\big]$ only contributes to the real part of $\kk$, and perturbations of the imaginary part appear in the next order through the asymptotic freedom $\kk_2$, as shown in Fig.\,\ref{fig:falloff:Imkk}. Finally $\Nhat_0=0$ and $\Nhat_1=\kkappa_0^{-1}$ is determined by the background. Perturbations of $\Nhat$ appear in the second order, through $\Nhat_2=-\kkappa_0^{-2}\KK_1/2$. The plots obtained with the numerical data to support this claim are shown in Fig.\,\ref{fig:falloff:Nh}.
  \begin{figure}
  \centering
    \begin{subfigure}{.45\textwidth}
    \centering
      \includegraphics[width=\textwidth]{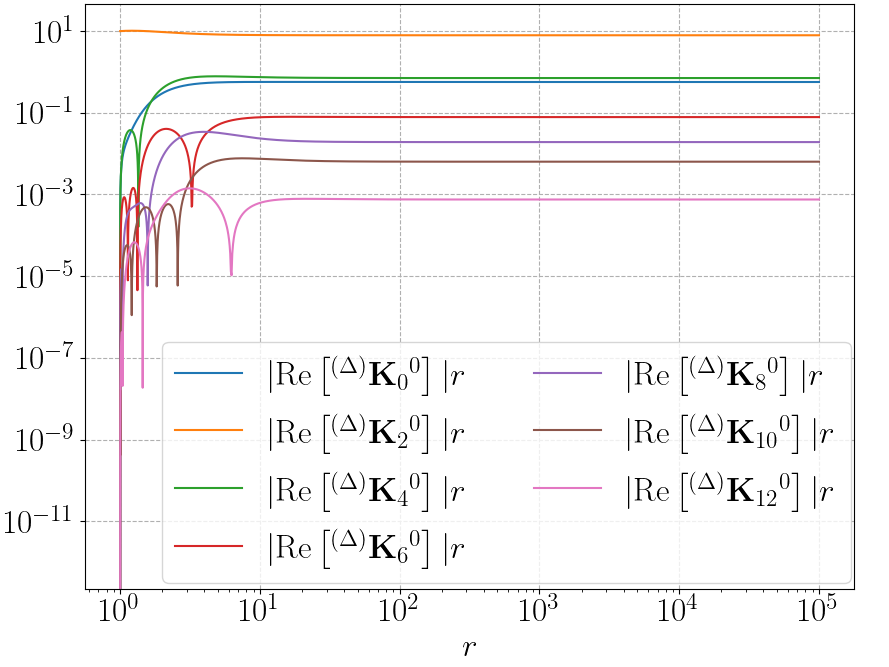}
      \caption{Asymptotic behavior of the product of the dominant modes of $\dKK$ and $r$.}
      \label{fig:falloff:KK}
    \end{subfigure}\hspace{.5cm}
    \begin{subfigure}{.45\textwidth}
    \centering
      \includegraphics[width=\textwidth]{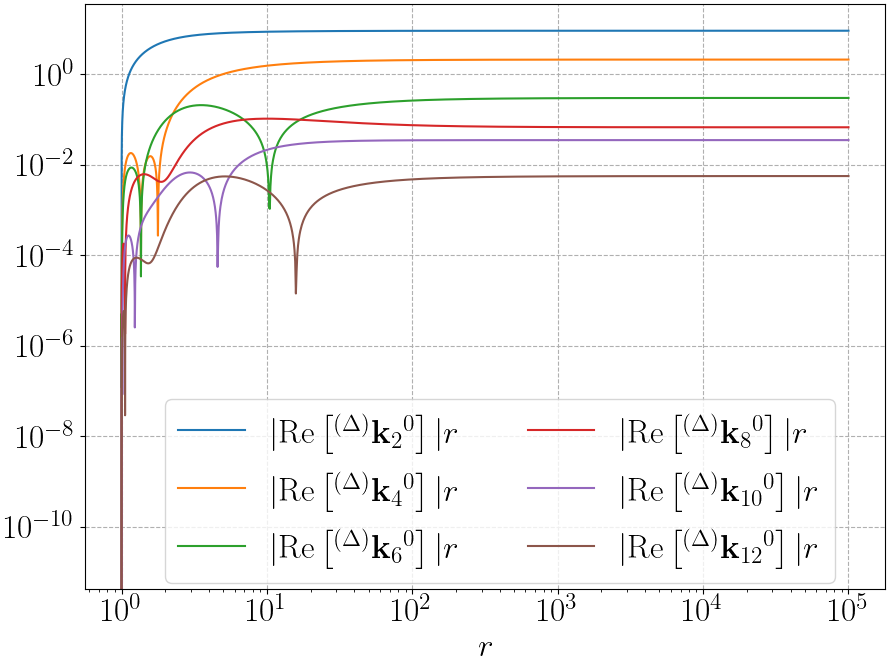}
      \caption{Asymptotic behavior of the product of the dominant modes of $\mathfrak{Re}\dkk$ and $r$.}
      \label{fig:falloff:Rekk}
    \end{subfigure}
    \\
    \begin{subfigure}{.45\textwidth}
    \centering
      \includegraphics[width=\textwidth]{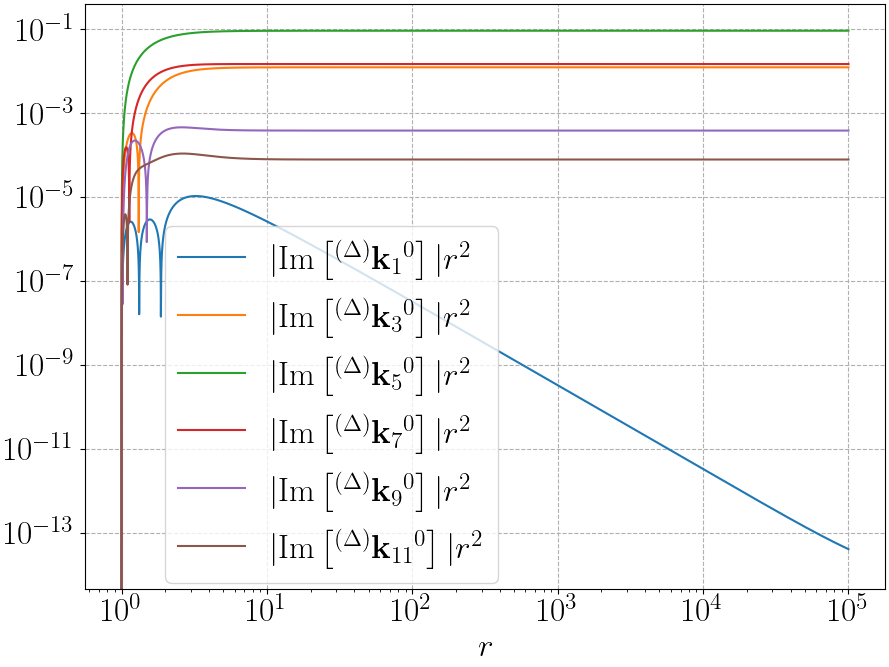}
      \caption{Asymptotic behavior of the product of the dominant modes of $\mathfrak{Im}\dkk$ and $r^2$.}
      \label{fig:falloff:Imkk}
    \end{subfigure}\hspace{.5cm}
    \begin{subfigure}{.45\textwidth}
    \centering
      \includegraphics[width=\textwidth]{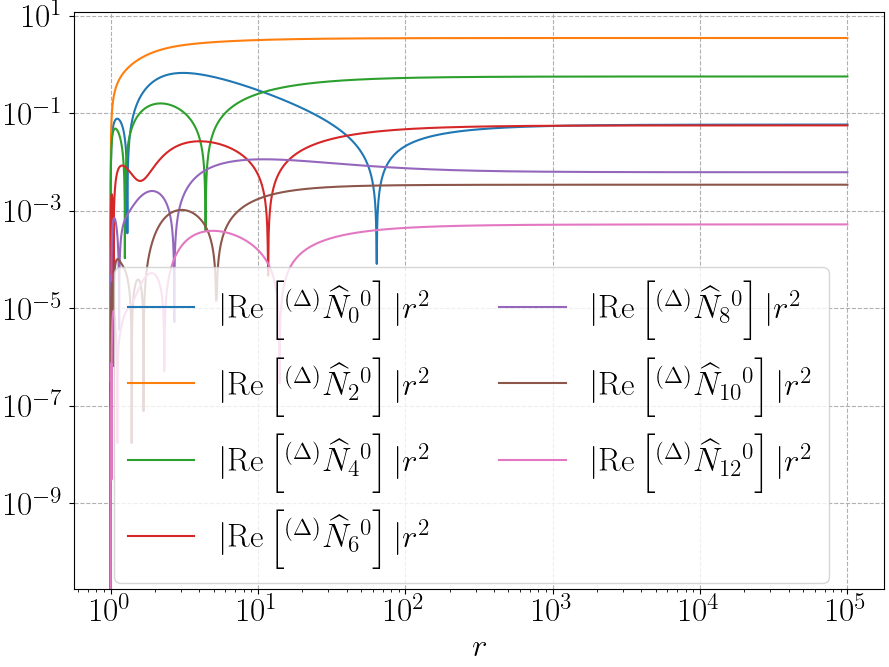}
      \caption{Asymptotic behavior of the product of the dominant modes of $\dNh$ and $r^2$.}
      \label{fig:falloff:Nh}
    \end{subfigure}
    \caption{The asymptotic behavior of the dominant modes of the constrained variables is indicated by the log-log plots through five orders of magnitude of the radial coordinate. Instead of plotting the radial dependence of the modes themselves, we plot these modes multiplied by the expected power of the radial coordinate, which in particular also demonstrates the accuracy level of the applied numerical integrator.}
    \label{fig:falloff}
  \end{figure}
  Summarizing the behavior shown we also get 
  \begin{align}
    \KK&=\bgKK_0+(\bgKK_1+\dKK_1)r^{-1}+\OO{-2}\,,\label{eq:asymK}\\
    \kk&=\bgkk_0+(\bgkk_1+\mathfrak{Re}[\dkk]_1)r^{-1}+\OO{-2}\,,\label{eq:asymk}\\
    \Nhat&=\bgNh_1r^{-1}+(\bgNh_2+\dNh_2)r^{-2}+\OO{-3}\,.\label{eq:asymNh}
  \end{align}
  The property that perturbations appear only in the next to leading order is exploited as follows. The known explicit dependence of the expansion coefficients on the asymptotic degrees of freedom can then be used to determine the falloff properties of various combinations of the constrained fields and the free data. For example, since the leading order behavior is completely determined by the free data, the following three relations can be seen to hold in a straightforward manner.

  The first one is $|\KK-2\kkappa_0|\sim r^{-1}$. In virtue of equation \eqref{eq:asymK} we get
  \begin{equation}
    |\KK-2\kkappa_0|=|\bgKK_0-2\kkappa_0+\OO{-1}|=\OO{-1},
  \end{equation}
  since the background satisfies $\bgKK_0=2\kkappa_0$.
  The second one is $|\kk-\eth\kkappa_0/\kkappa_0|\sim r^{-1}$. It follows from equation \eqref{eq:asymk} that
  \begin{equation}
    |\kk-\eth\kkappa_0/\kkappa_0|=|\bgkk_0-\eth\kkappa_0/\kkappa_0+\OO{-1}|=\OO{-1}\,,
  \end{equation}
  since the background satisfies $\bgkk_0=\eth\kkappa_0/\kkappa_0$.
  The third one is $|\Nhat-(\kkappa_0\,r)^{-1}|\sim r^{-2}$. Then, equation \eqref{eq:asymNh} implies
  \begin{equation}
    |\Nhat-(\kkappa_0\,r)^{-1}|=|(\bgNh_1-\kkappa_0{}^{-1})r^{-1}+\OO{-2}|=\OO{-2}\,,
  \end{equation}
  since the background satisfies $\bgNh_1=\kkappa_0{}^{-1}$.

  The last four interrelations involve somewhat more complicated combinations of the constraint variables, and thus provide a more rigorous test of the numerical accuracy achieved. The fourth interrelation is $|(\KK-4\kkappa_0)r^{-1}+2\kkappa_0{}^2\Nhat|\sim r^{-3}$. Splitting this into background and deviation gives us
  \begin{multline}
    |(\KK-4\kkappa_0)r^{-1}+2\kkappa_0{}^2\Nhat|=\\
    \left|\left[(\bgKK-4\kkappa_0)r^{-1}+2\kkappa_0{}^2\bgNh\right]+\dKK r^{-1}+2\kkappa_0{}^2\dNh\right|\,.
  \end{multline}
  The terms in square brackets are given only in terms of background variables, and the entire term belongs to $\OO{-3}$. The behavior of the remaining two terms can be seen on one of the plots in Fig.\,\ref{fig:smoothnessallinone}. We see that the deviation satisfies this relation with reasonable accuracy.
  The fifth one is $|\kk-\eth\kkappa_0/\kkappa_0-\frac{1}{2}\eth(\KK/\kkappa_0)|\sim r^{-2}$. Separating the background and the deviations we arrive at
  \begin{multline}
    \left|\kk-\eth\kkappa_0/\kkappa_0-\frac{1}{2}\eth(\KK/\kkappa_0)\right|=\\
    \left|\left[\bgkk-\eth\kkappa_0/\kkappa_0-\frac{1}{2}\eth(\bgKK/\kkappa_0)\right]+\dkk-\frac{1}{2}\eth(\dKK/\kkappa_0)\right|\,.
  \end{multline}
  The bracketed term again contains only background quantities and the whole belongs to $\OO{-2}$. The behavior of the sum of the remaining two terms is shown by one of the plots in Fig.\,\ref{fig:smoothnessallinone}. This verifies that $|\kk-\eth\kkappa_0/\kkappa_0-\frac{1}{2}\eth(\KK/\kkappa_0)|$ decays at the expected rate.

  The sixth relation is $|\eth(\kkappa_0\Nhat)+(\kk-\eth\kkappa_0/\kkappa_0)r^{-1}|\sim r^{-3}$. Separating the background and the deviation terms as before, we get
  \begin{multline}
    |\eth(\kkappa_0\Nhat)+(\kk-\eth\kkappa_0/\kkappa_0)r^{-1}|=\\
    \left|\left[\eth(\kkappa_0\bgNh)+(\bgkk-\eth\kkappa_0/\kkappa_0)r^{-1}\right]+\eth(\kkappa_0\dNh)+\dkk r^{-1}\right|\,.
  \end{multline}
  The terms in parentheses are only background variables and belong to $\OO{-3}$. The behavior of the remaining terms is shown by one of the plots in Fig.\,\ref{fig:smoothnessallinone}, which verifies the expected decay rate.

  The seventh relation is $|\eth(\kkappa_0\Nhat)+\frac{1}{2}\eth(\KK/\kkappa_0)r^{-1}|\sim r^{-3}$. Separating the background and the deviations we get 
  \begin{multline}
    \left|\eth(\kkappa_0\Nhat)+\frac{1}{2}\eth(\KK/\kkappa_0)r^{-1}\right|=\\
    \left|\left[\eth(\kkappa_0\bgNh)+\frac{1}{2}\eth(\bgKK/\kkappa_0)r^{-1}\right]+\eth(\kkappa_0\dNh)+\frac{1}{2}\eth(\dKK/\kkappa_0)r^{-1}\right|\,.
  \end{multline}
  The term in parentheses contains only background variables and is $\OO{-3}$. The behavior of the additional terms is shown by one of the plots in Fig.\,\ref{fig:smoothnessallinone}. We see that the deviation satisfies this relation with reasonable accuracy.

  \begin{figure}[H]
  \centering
    \includegraphics[width=.8\textwidth]{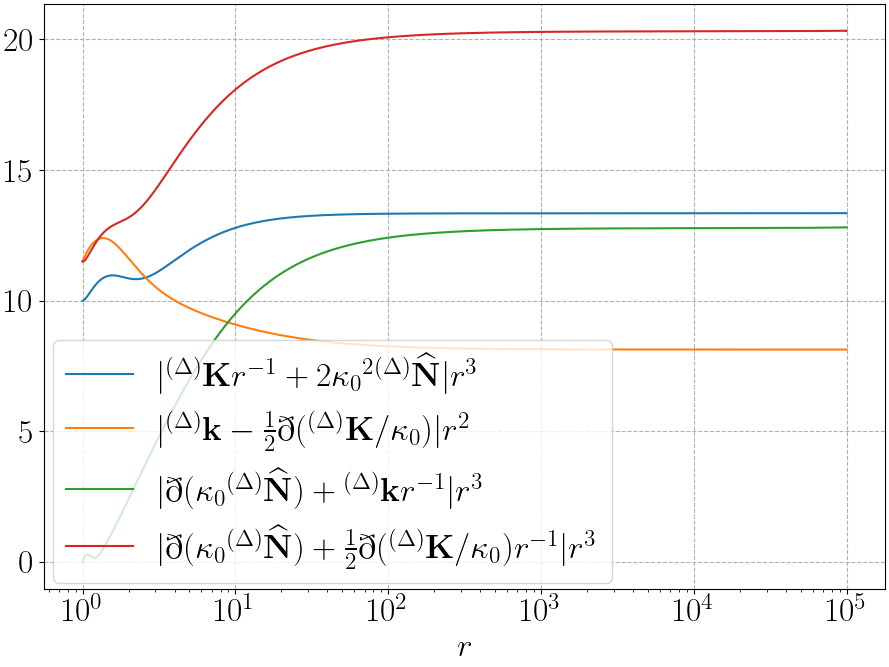}
    \caption{The asymptotic behavior of various combinations of the constrained fields and the free part of the data are plotted. Similar to Fig.\,\ref{fig:falloff}, instead of plotting the radial dependence of the modes themselves, we plot these modes multiplied by the expected power of the radial coordinate, which in particular also demonstrates the accuracy level of the applied numerical integrator. All of these plots confirm that the relevant combinations decay at the expected rates, as discussed in the text above. }
    \label{fig:smoothnessallinone}
  \end{figure}

  \subsection{Numerical evaluation of the Hawking and Bondi masses}
  \label{subsec: approx Hawking}
  
  The proof of Theorem \ref{main theorem 1} already indicated that a numerical evaluation of the Hawking mass can be tricky. This is so because if the desired cancellations of the higher order terms do not occur then the numerical evaluation of the Hawking mass will be contaminated by these errors introduced by the lower order terms.\,\footnote{Note that we do not have the analogous problem in the asymptotically flat case, since $\Nhat\rightarrow 1$ in the asymptotic limit. The asymptotic expansion of the critical integrand is then written as
  	\begin{align}
  		\Theta^{(+)}\Theta^{(-)}\sqrt{\dd} = \, & \KK_0^2\cdot r^2+2 \KK_0 \KK_1 \cdot r
  		+\left(\aaa_0 \KK_0^2+2 \KK_0 \KK_2+\KK_1^2-4\Nhat_0^{-2}\right)+\OO{-1},
  	\end{align}
  	with $\Nhat_0=1$, $\KK_0=\KK_1=0$, which automatically gives the desired asymptotic form $-4+\OO{-1}$. Of course, in this case the Hawking mass asymptotes to the ADM mass.}
  
  To remedy these difficulties, motivated by ideas used in \cite{Beyer:2021kmi}, we used the radial variation of the Hawking mass given by \eqref{eq:Hexp}, which can be written as
  \begin{equation}
  	\mathscr{L}_r m_H=\frac{1}{2}m_H\frac{\mathscr{L}_r\mathcal{A}}{\mathcal{A}}+\frac{1}{16\pi}\sqrt{\frac{\mathcal{A}}{16\pi}}\int_{\mathscr{S}_r}\mathscr{L}_r\left(\big[\KK^2-\Kh^2\big]\,\epshat\right)\,,
  \end{equation}
  where $\Theta^{(+)}\Theta^{(-)}=\KK^2-\Kh^2$ was used. Then, using the relations 
  \begin{equation}
  	\mathscr{L}_r\mathcal{A}=\int_{\mathscr{S}_r}\Kstar\,\epshat\,,\quad {\rm and} \quad \mathscr{L}_r\,\epshat= \big[\Nhat\Kh+(\Dhat_i\Nhat^i)\big]\,\epshat\,,
  \end{equation}
  together with $\Kstar=\Nhat\Kh$ and $r^a=\Nhat n^i + \Nhat^i$ \cite{Racz:2017krc}, we get  
  \begin{align}\label{eq:var0}
  	\int_{\mathscr{S}_r}\mathscr{L}_r\left(\big[\KK^2-\Kh^2\big]\,\epshat \right)
    &=\int_{\mathscr{S}_r}\left(2\KK\mathscr{L}_r\KK-2\Kh\mathscr{L}_r\Kh\right)\epshat+\int_{\mathscr{S}_r}\big[\KK^2-\Kh^2\big]\,\left(\Kstar+(\Dhat_i\Nhat^i)\right)\,\epshat\,,
  \end{align}
  Noticing then that $\mathscr{L}_r\Kh=\mathscr{L}_r(\Kstar \Nhat^{-1})=\mathscr{L}_r(\Kstar) \Nhat^{-1}-\Kstar\Nhat^{-2} \mathscr{L}_r(\Nhat)$ we have that the integrands in \eqref{eq:var0} can be given by the free data and by the terms $\mathscr{L}_r(\Nhat)$ and $\mathscr{L}_r\KK$ determined by the parabolic-hyperbolic equations \eqref{eq:phN} and \eqref{eq:phK}, respectively. Based on this, and dropping total divergences whose integral does not contribute, we get 
  \begin{align}
    \mathscr{L}_r m_H&=\frac{1}{2}\frac{m_H}{\mathcal{A}}\int_{\mathscr{S}_r}\Kstar\,\epshat+\frac{1}{16\pi}\sqrt{\frac{\mathcal{A}}{16\pi}}\int_{\mathscr{S}_r}\left(2\KK\,\partial_r\KK-2\Kstar\Nhat^{-2}\big[\mathscr{L}_{r}\Kstar -\Nhat^{-1} \mathscr{L}_{r}\Nhat \big])\right)\epshat\nonumber\\
    &+\frac{1}{16\pi}\sqrt{\frac{\mathcal{A}}{16\pi}}\int_{\mathscr{S}_r}\left(\KK^2-\Kh^2\right)\Kstar\,\epshat\nonumber\\
    &-\frac{1}{16\pi}\sqrt{\frac{\mathcal{A}}{16\pi}}\int_{\mathscr{S}_r}\left[\KK\,(\NNt\,\ethb\KK+\NNtbar\,\eth\KK)-\Kh\,(\NNt\,\ethb\Kh+\NNtbar\,\eth\Kh)\right]\epshat\,.
    \label{eq:drmH}
  \end{align}
  Note that although \eqref{eq:drmH} explicitly includes $m_H$, this is not a technical problem, since \eqref{eq:drmH} is of first order and requires the value of $m_H$ given by \eqref{eq:Hexp} only on the initial surface. The value of $m_H$ on other $r=const$ level surfaces is determined by integration. 

  Although the procedure outlined above was inspired by \cite{Beyer:2021kmi}, there are some important differences. The authors in \cite{Beyer:2021kmi} kept only the leading order terms from the asymptotic expansions $\sqrt{\dd}=r^2+\OO{0}$ and $\sqrt{\mathcal{A}/16\,\pi}=r/2+\OO{0}$. Accordingly, in \cite{Beyer:2021kmi} the evolution equation for the Hawking mass was derived from
  \begin{equation}
  	m_H=\frac{r}{2}+\frac{r^3}{32\pi}\int_{\mathscr{S}_r}(\KK^2-\Kstar{}^2\Nhat{}^{-2})\, \boldsymbol{\epsilon}_q\,,
  	\label{eq:Hleading}
  \end{equation}
  where $\boldsymbol{\epsilon}_q$ is the volume element of the unit sphere. Note that this formula is only exact if the free data $\widehat{\gamma}_{ab}$ is spherically symmetric. However, as shown on Fig.\,\ref{fig:HawkingComp}, the Hawking mass determined using \eqref{eq:Hleading} is insensitive to nonlinear contributions from the strong field regime. On Fig.\,\ref{fig:HawkingComp} are the plots obtained by integrating \eqref{eq:drmH}, with the Hawking mass evaluated by integrating (3.48) of \cite{Beyer:2021kmi}, and with the sphere-by-sphere evaluation of \eqref{eq:Hexp}. While the first two methods produce stable asymptotic values, the sphere-by-sphere evaluation of \eqref{eq:Hexp} quickly leads to inaccurate values due to rounding errors. The two asymptotic values differ significantly because the integration of (3.48) of \cite{Beyer:2021kmi} is insensitive to the strong field region contributions in the Kerr case. Note that the similarities in the initial behavior of the sphere-by-sphere evaluation of \eqref{eq:Hexp} and that obtained by the integration of \eqref{eq:drmH} are convincing. The Hawking mass tends to an asymptotic value surprisingly fast, but at the very end, after $r=10^4$, we lose precision and numerical noise dominates the data. A polynomial fit ignoring the last $100$ points (end of fit around $r=17500$) tells us that the specific value of the Bondi mass is $m_B=1.868$, almost twice the ``Bondi mass'' of the Kerr background, which also confirms that the applied perturbation is highly non-linear.
  \begin{figure}[H]
  \centering
  \begin{subfigure}{.465\textwidth}
    \centering
    \includegraphics[width=.9\textwidth]{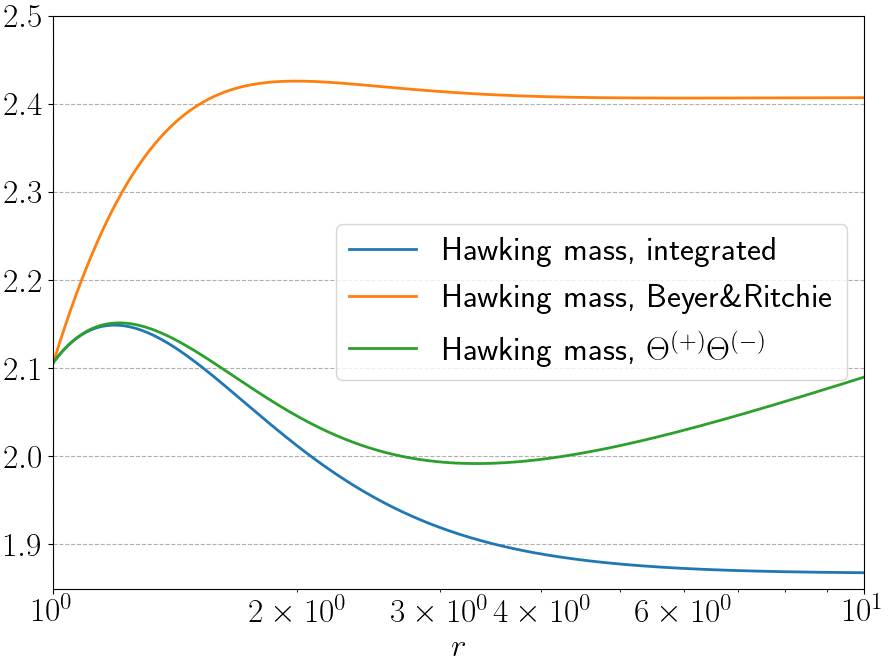}
    \caption{The radial dependence of the Hawking mass is shown using the three different methods described in the text above.}
    \label{fig:HawkingComp}
  \end{subfigure}\hspace{.5cm}
  \begin{subfigure}{.465\textwidth}
    \centering
    \includegraphics[width=.9\textwidth]{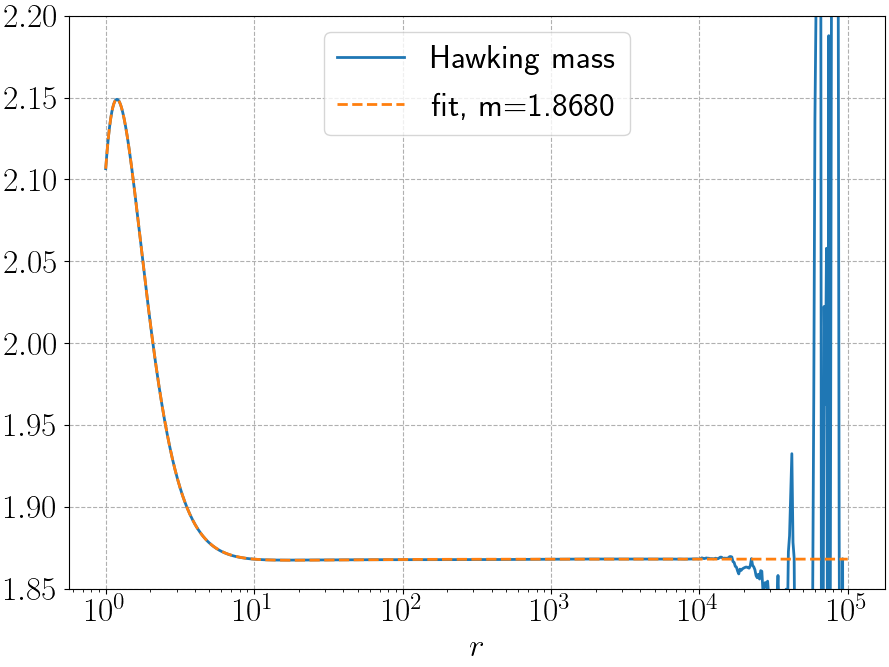}
    \caption{The asymptotic behavior of the Hawking mass, obtained by integrating \eqref{eq:drmH}, is shown along with a polynomial fit.}
    \label{fig:HawkingWhole}
  \end{subfigure}
  \caption{The radial dependence of the Hawking mass of the perturbed initial data is plotted. The left panel compares the result obtained by three different methods as discussed in the text above. The blue graph shows the result obtained by integrating \eqref{eq:drmH}, the orange graph shows the behavior determined by the simplified formula used in \cite{Beyer:2021kmi}, and finally the green graph is obtained by a simple evaluation of the integral in \eqref{eq:Hexp}. The right panel shows the asymptotic behavior of the Hawking mass resulting from the integration of \eqref{eq:drmH}, along with a polynomial fit that was used to read the value of the Bondi mass.}
  \label{fig:Hawking}
  \end{figure}

  \section{Summary}
  \label{sec:summ}
  This paper consists of two essentially separate parts. The dominant part of the paper provides a careful analysis of analytic issues related to the space of generic properties of asymptotically hyperboloidal initial data configurations. These investigations were motivated by the fact that the asymptotic expansions of generic asymptotically hyperboloidal initial data solutions to the elliptic form of the constraint equations are long known to involve polyhomogeneous functions \cite{Andersson:1992yk,Andersson:1993we,Andersson:1994ng,Andersson:1996xd}.
  
  Our first main result is that once the considerations are restricted to the asymptotically hyperboloidal initial data subject to the parabolic-hyperbolic form of the constraint equations which also admit a  well-defined Bondi mass the corresponding asymptotically hyperboloidal initial data sets are free of logarithmic singularities. Note that the essential difference between the evolutionary and elliptic forms of the constraints is that the Lichnerowicz-York method makes extensive use of conformal rescaling of the constrained variables, whereas no such conformal rescaling occurs when the parabolic-hyperbolic form of the constraints is used. It seems plausible that this key difference accounts for the efficiency of the evolutionary method in the present case.
  
  Our second main result provides a substantial generalization of a recent result of Beyer and Ritchie \cite{Beyer:2021kmi} by demonstrating that the existence of well-defined Bondi mass and angular momentum, together with some mild consistently solvable algebraic restrictions on a part of the free data, implies that the generic solutions of the parabolic-hyperbolic form of the constraint equations are smooth, thereby entirely free of logarithmic singularities. Combining these results with those of the corresponding hyperboloidal initial value problem \cite{Friedrich:1991nn,Friedrich:2015bax,Frauendiener:2000mk,Kroon:2016ink} we can conclude that the Cauchy developments of the corresponding asymptotically hyperboloidal initial data specifications must admit smooth conformal boundary as assumed in the original definition of Penrose \cite{Penrose:1962ij}.
  
  \medskip
  
  The analytical study discussed above is complemented by some related numerical investigations. In particular, we show that highly nonlinear perturbations of near Kerr initial data can be constructed numerically. It is also shown that our analytical results clearly manifest themselves in these precise numerical studies. For example, various combinations of the numerical data fall off exactly at the rate that can be deduced from our analytic considerations. A careful numerical evaluation of the Hawking mass of the resulting initial data, together with the asymptotic behavior of the Hawking mass, and thus the value of the Bondi mass, is also determined. Finally, note that all the conditions on the free data, given by \eqref{eq:alg-theorem}, \eqref{eq:extra2} and \eqref{eq:extra1}, automatically hold for the considered nonlinear perturbations of the Kerr initial data because the free data, as given in \eqref{eqs: bgNh}-\eqref{eqs: Kcqqbg}, satisfies them.

  \section*{Acknowledgments}

  The authors would like to thank Ingemar Bengtsson, Piotr Chrusciel, Gábor Tóth, and Jeffrey Winicour for their helpful comments.
  This work was supported in part by the Hungarian Scientific Research Fund NKFIH Grant No. K-142423 and NSF CAREER Award PHY-2047382.
  Data presented in this paper was produced on the Maple cluster of Mississippi Center for Supercomputing Research.

  \section*{Software and data availability}

  Data was produced by \verb|constraintSolver| \cite{noauthor_csukas_2024} available under MIT license. Figures were produced using \verb|matplotlib| \cite{Hunter:2007}, \verb|numpy| \cite{harris2020array}, and \verb|pandas| \cite{reback2020pandas,mckinney-proc-scipy-2010}. The raw data along with scripts producing the figures are available to the public \cite{csukas_2025_14779011}.
  
  \appendix
  \numberwithin{equation}{section}
  \setcounter{equation}{0}
  
  
  \section{Terms involving the two-metric ${\widehat \gamma}{}_{kl}$}

  The metric that is induced on the $\mathscr{S}_r$ level surfaces can be given as 
  \begin{equation}\label{ind_metr}
  	\widehat\gamma_{ab}={\bf a}\, q_{ab}+\tfrac12\left[ \mathbbm{\bf b} \,\overline q_a \overline q_b + \,\overline{\mathbbm{\bf b}}\, q_a q_b \right]\,, 
  \end{equation}
  where 
  \begin{equation}
  	\mathbbm{\bf a}=\tfrac12\,\widehat\gamma_{ab}\,q^a \,\overline q^b \hskip0.5cm  {\rm and} \hskip0.5cm \mathbbm{\bf b}=\tfrac12\,\widehat\gamma_{ab}\,q^a q^b\,.
  \end{equation}
  It is straightforward to verify that the inverse $\widehat\gamma^{ab}$ metric can be written as
  \begin{equation}\label{inv_ind_metr}
  	\widehat\gamma^{ab}=\mathbbm{\bf d}^{-1}\left\{\mathbbm{\bf a}\, q^{ab}-\tfrac12\left[ \mathbbm{\bf b} \,\overline q^a \overline q^b
  	+ \,\overline{\mathbbm{\bf b}}\, q^a q^b \right]\right\}\,, 
  \end{equation}
  where 
  \begin{equation}
  	\mathbbm{\bf d}=\mathbbm{\bf a}^2-\mathbbm{\bf b}\,\overline{\mathbbm{\bf b}}
  \end{equation}
  denotes for the ratio  $\det(\widehat\gamma_{ab})/\det(q_{ab})$ of the determinants of $\widehat\gamma_{ab}$ and $q_{ab}$, which is a function on the $\mathscr{S}_r$ level surfaces.

  \section{The trace-free part of ${\widehat K}{}_{kl}$}
   
  The extrinsic curvature $\widehat K_{ij}$ of $\mathscr{S}_r$ is given as
  \begin{align}\label{hatextcurv2}
  	\widehat K_{ij}=  \tfrac12\,\mathscr{L}_{\widehat n} {\widehat \gamma}_{ij}={}& \tfrac12\,{\widehat N}^{-1}[\,\mathscr{L}_{r}{\widehat \gamma}_{ij}
  	- ( \widehat D_i\widehat N_j + \widehat D_j\widehat N_i )] \nonumber \\ = {}& \tfrac12\,{\,\widehat{N}}^{-1}[(\partial_r\mathbbm{\bf a})\,q_{ij}
  	+\tfrac12\,[\left(\partial_r\mathbbm{\bf b}\right)\,\overline q_i\,\overline q_j + \left(\partial_r\,\overline{\mathbbm{\bf b}}\right) q_i q_j] 
  	- (\widehat D_i\widehat N_j + \widehat D_j\widehat N_i )]\,,
  \end{align}
  where in the last step the Lie-invariance of the dyad field $q^i$ and the explicit form of (\ref{ind_metr}) were used.  We also get
  \begin{align}\label{hatextcurv_trace}
  	{}& \,\widehat{{K}}  = \tfrac12\,({\widehat{{N}}\,\mathbbm{\bf d}})^{-1} \biggl[\Bigl\{\mathbbm{\bf a}\,[\,(\partial_r\mathbbm{\bf a}) - (\,\overline\eth\,\mathbf{N}) 
  	+ \,\overline{\mathbf{B}}\,\mathbf{N} \,] \Bigr. \nonumber 
  	\\ {}& \hskip4.1cm \Bigl.  - \mathbbm{\bf b}\,[\,(\partial_r\,\overline{\mathbbm{\bf b}}) - \,(\,\overline{\eth}\,\overline{\mathbf{N}}) 
  	+\tfrac12\, \,\overline{\mathbf{C}}\,\mathbf{N} +\tfrac12\, \,\overline{\mathbf{A}}\,\overline{\mathbf{N}} \,]  \Bigr\} + ``\,cc\,"\biggr] \, ,
  \end{align}
  and using \eqref{hatextcurv2} that
  \begin{equation}\label{hatextcurv_qq}
  	\,\widehat{{K}}{}_{qq} =  q^i q^j\widehat K_{ij} 
  	=\tfrac12{\widehat{{N}}}^{-1}\left\{2\,\partial_r\mathbbm{\bf b} - 2\,\eth\,\mathbf{N}
  	+ {\mathbf{C}}\,\overline {\mathbf{N}} +\mathbf{A} {\mathbf{N}} \,  \right\} \,, 
  \end{equation}
    where 
    \begin{align} \label{ABC}
    	\mathbf{A} = {}&  q^a q^b {C^e}{}_{ab}\,\overline q_e = \mathbf{d}^{-1}\left\{ \mathbf{a}\left[2\,\eth\,\mathbf{a} -\,\overline{\eth}\,\mathbf{b}\right] 
    	-  \,\overline{\mathbf{b}}\,\eth\,\mathbf{b} \right\} \nonumber \\
    	\mathbf{B} = {}&  \,\overline q^a q^b {C^e}{}_{ab}\,q_e = \mathbf{d}^{-1}\left\{ \mathbf{a}\,\overline{\eth}\,\mathbf{b}
    	- \mathbf{b}  \,\eth\,\overline{\mathbf{b}}\right\} \\
    	\mathbf{C} = {}&  q^a q^b {C^e}{}_{ab}\,q_e = \mathbf{d}^{-1}\left\{ \mathbf{a}\,\eth\,\mathbf{b} -  \mathbf{b}\left[2\,\eth\,\mathbf{a}
    	-\,\overline{\eth}\,\mathbf{b}\right] \right\} \,, \nonumber
    \end{align}
    where the $(1,2)$-type tensor field ${C^e}{}_{ab}$ relates the covariant derivative operators $\widehat D_{a}$ and ${\mathbb D}_{a}$ with respect to the metrics, $\widehat{\gamma}_{ab}$ and $q_{ab}$, respectively, (hence, in particular, $\widehat D_{i}\widehat N_j={\mathbb D}_{i} \widehat N_j-{C^k}{}_{ij}\widehat N_k$ holds), and also where $``\,cc\,"$ abbreviates the complex conjugate of the preceding term at the appropriate level of the hierarchy.

    Then, using the trace free part of ${\widehat K}_{ij}$ 
    \begin{equation}\label{eq: intKhat}
    	\interior{\widehat K}_{ij}={\widehat K}_{ij}-\tfrac12\,\widehat \gamma_{ij}{\widehat K}\,,
    \end{equation}
    we get
    \begin{equation}\label{eq: intKhat2}
    	\interior{\widehat K}_{qq}= {\widehat K}_{qq} - \mathbbm{\bf b} \, {\widehat K}\,.
    \end{equation}
    which along with \eqref{ind_metr}, \eqref{hatextcurv_trace} and \eqref{eq: intKhat2} gives finally that
  	\begin{align}\label{inthatextcurv_qq}
  	\,\interior{\widehat{K}}{}_{qq} = {}& 
 	 \tfrac12{\widehat{{N}}}^{-1}\Bigl\{2\,\partial_r\mathbbm{\bf b} - 2\,\eth\,\mathbf{N}
   	+ {\mathbf{C}}\,\overline {\mathbf{N}} +\mathbf{A} {\mathbf{N}} \,  \Bigr\} \nonumber  \\ 
   	{}& + 
   	\mathbbm{\bf b}\,({\widehat{{N}}\,\mathbbm{\bf d}})^{-1} \biggl[ \Bigl\{\mathbbm{\bf a}\,\bigl[\,(\partial_r\mathbbm{\bf a}) - (\,\overline\eth\,\mathbf{N})  + \,\overline{\mathbf{B}}\,\mathbf{N} \,\bigr] \Bigr.\biggr. \nonumber \\ {}& \hskip2.5cm \biggl. \Bigl. - \mathbbm{\bf b}\,\bigl[\,(\partial_r\,\overline{\mathbbm{\bf b}}) - \,(\,\overline{\eth}\,\overline{\mathbf{N}}) +\tfrac12\, \,\overline{\mathbf{C}}\,\mathbf{N} +\tfrac12\, \,\overline{\mathbf{A}}\,\overline{\mathbf{N}} \,\bigr]  \Bigr\} + ``\,cc\,"
   	\biggr] \,. 
 	 \end{align}
 	 
 	Note also that for any symmetric trace-free tensor field $\interior{T}_{ab}$ the algebraic contents of $\interior{T}_{ab}$ and $\interior{T}_{qq}=\interior{T}_{ab}q^aq^b$ are equivalent. To see this, recall that  
 	 \begin{align}\label{eq: intTqq}
 	 	\interior{T}_{ab}=\tfrac14\,\interior{T}_{ef}(q_a \overline q^e + \overline q_a q^e)(q_b \overline q^f + \overline q_b q^f) = \tfrac14\,\big[q_aq_b \interior{T}_{\overline q \,\overline q}+2\,q_{ab}\interior{T}_{q \overline q} + \overline q_a\overline q_b \interior{T}_{q q}\big]\,
 	 \end{align}
	holds, where 
	\begin{equation}
		\interior{T}_{\overline q \,\overline q}=\overline{\interior{T}_{q q}}\quad {\rm and}\quad \interior{T}_{q \,\overline q}=(2\,\aaa)^{-1}\big[\,\bb\,\overline{\interior{T}_{q q}}+\overline{\bb}\,\interior{T}_{q q}\big]\,.
	\end{equation}

	\section{Useful relations for conformal rescaling operations}\label{appendix:decomp}
	
	\begin{itemize}
		\item ${\widetilde n}_a=\Omega\,n_a$\,, \ 
		 ${\widetilde h}_{ab}=\Omega^2\,h_{ab}$\,, \ ${\widetilde h}^{ab}=\Omega^{-2}\,h^{ab}$ \ \ projectors \ \ ${\widetilde h}_{a}{^b}= {\widetilde h}_{ae}{\widetilde h}^{eb}= h_{ae}h^{eb}= h_{a}{^b}$
		\item ${\widetilde K}_{ab}=\Omega\,K_{ab} - \mathscr{L}_{n} \Omega\,{h}_{ab} $
		\item ${\tildeon \nhat}_a=\Omega\,\nhat_a$\,, \  ${\tildeon \gammahat}_{ab}=\Omega^2\,\gammahat_{ab}$\,, \ ${\tildeon \gammahat}{}^{ab}=\Omega^{-2}\,\gammahat^{ab}$ \ \ projectors \ \ ${\tildeon \gammahat}_{a}{^b}= {\tildeon \gammahat}_{ae}{\tildeon \gammahat}{}^{eb}= \gammahat_{ae}\gammahat^{eb}= \gammahat_{a}{^b}$
		\item ${\tildeon \KK}_{ab}=\Omega\,\KK_{ab} - \mathscr{L}_n \Omega\,\gammahat_{ab}$\,, \ 
		 $\circon{\tildeon \KK}_{ab}=\Omega\,\circon\KK_{ab}$
		\item ${\tildeon{\widehat K}}_{ab}=\Omega\,{\widehat K}_{ab} - \mathscr{L}_{\widehat{n}} \Omega\,\gammahat_{ab}$\,, \ 
		${{\tildeon {\widehat K}}}{}^\circ_{ab}=\Omega\,\circon{\widehat K}_{ab}$
		
		\item  Some components of the conformally rescaled (non-physical) extrinsic curvature are
			\begin{align}	
			\hskip-3.2cm\widetilde{K}_{rr} & = \Omega\,K_{rr} - \mathscr{L}_n \Omega\,\big({\widehat N}^2+{\widehat N}_A {\widehat N}^A \big)\\  & = \Omega\,\big[{\widehat N}^2 \kkappa  + 2\,{\widehat N} \kk_A {\widehat N}^A + \KK_{AB}{\widehat N}^A{\widehat N}^B\big] - \,\mathscr{L}_n \Omega\,\big({\widehat N}^2+{\widehat N}_A {\widehat N}^A \big) \nonumber\\
			\widetilde{K}_{r A} & = \Omega\,\big[{\widehat N} \kk_A  + \KK_{AB}\gammahat^{BF}{\widehat N}_F\big] - \mathscr{L}_n \Omega\,{\widehat N}_A \,,
		    \end{align}
			where the $(2+1)$-decomposition of the three-metric
				\begin{equation}
					h_{ab}=\left(
					\begin{array}{cc}
								{\widehat N}^2+{\widehat N}_A {\widehat N}^A  & {\widehat N}_A \\
								{\widehat N}_B & \gammahat_{AB}
					\end{array}
					\right)\,,
		    	\end{equation}
	       and the reconstruction of the physical extrinsic curvature 
	    	\begin{enumerate}
	    		\item[(i)] $K_{rr} = {\widehat N}^2 \kkappa  + 2\,{\widehat N} \kk_A {\widehat N}^A + \KK_{AB} {\widehat N}^A{\widehat N}^B $
	    		\item[(ii)] $K_{r A} = {\widehat N} \kk_A  + \KK_{AB}{\widehat N}^B $
	    		\item[(iii)]  $K_{AB} = \KK_{AB}\,,$
	        \end{enumerate}
            were used.
		    
		    Finally, some contractions
            \begin{enumerate}
            	\item[(i)] 
            	\begin{equation}\label{appendix:gammaNN}
            		\gammahat_{EF}{\widehat N}^E{\widehat N}^F={\widehat N}_E {\widehat N}^E =\tfrac12\,(q_F \overline q^E + \overline q_F q^E) {\widehat N}_E {\widehat N}^F=\tfrac12\,({\mathbf{N}} \overline{\tildeon {\mathbf{N}}}+\overline{\mathbf{N}} \tildeon {\mathbf{N}}) 
            	\end{equation}
            	\item[(ii)] 
            	\begin{equation}
            		{\mathbf{k}}_A=\tfrac12\,(q_A \overline q^E + \overline q_A q^E)\,{\mathbf{k}}_E= \tfrac12\,(q_A\overline{\mathbf{k}}+\overline q_A{\mathbf{k}})
            	\end{equation}
            	\item[(iii)] 
            	\begin{equation}
            		\KK_{AB}= \interior{\KK}_{AB} + \tfrac12\,\gammahat_{AB}\,\KK
            	\end{equation} 
            	\item[(iv)]    
            	\begin{align} 
            		\interior{\KK}_{AF}{\widehat N}^F & =\tfrac12\,(q_A \overline q^E + \overline q_A q^E)\,\interior{\KK}_{EF}{\widehat N}^F=\tfrac12\,(q_A \interior{\KK}_{{\overline q}F} + \overline q_A \interior{\KK}_{qF})\,{\widehat N}^F \\ \nonumber & = \tfrac12\,(q_A \interior{\KK}_{{\overline q}F} + \overline q_A \interior{\KK}_{qF})\,\tfrac12\,\big(q_H \overline q^F + \overline q_H q^F\big)\,\,{\widehat N}^H \\  \nonumber & =\tfrac14\,\big(q_A \big[\interior{\KK}_{{\overline q}{\overline q}}{\tildeon {\mathbf{N}}}+\interior{\KK}_{{ q}{\overline q}}\overline{\tildeon {\mathbf{N}}}\big] + \overline q_A \big[\interior{\KK}_{{ q}{\overline q}}\,{\tildeon {\mathbf{N}}}+\interior{\KK}_{{ q}{ q}}\,\overline{\tildeon {\mathbf{N}}}\big]\big) \nonumber \\ 
            		\tfrac12\,\gammahat_{AF}\,{\widehat {N}}^{F} & = \tfrac14\,\big[\aaa\, (q_A \overline{\tildeon {\mathbf{N}}} + \overline q_A {\tildeon {\mathbf{N}}})+(\overline q_A \bb\, \overline{\tildeon {\mathbf{N}}} + q_A \overline \bb \,{\tildeon {\mathbf{N}}})\big]
             	\end{align}
            	\item[(v)] 
            	\begin{align} \label{appendix:KNN}
            		\KK_{AB}{\widehat N}^A{\widehat N}^B & = \tfrac14\,\big(q_E \overline q^A + \overline q_E q^A\big)\big(q_F \overline q^B + \overline q_F q^B\big)\,\KK_{AB}{\widehat N}^E{\widehat N}^F  \\ \nonumber & = \tfrac14\,\big(\interior{\KK}_{{\overline q}{\overline q}}{\tildeon {\mathbf{N}}}^2 +  2\,\interior{\KK}_{{ q}{\overline q}}{\tildeon {\mathbf{N}}}\overline{\tildeon {\mathbf{N}}} + \interior{\KK}_{{ q}{ q}}\overline{\tildeon {\mathbf{N}}}^2\big)\nonumber
            		\end{align}
            \end{enumerate}
	\end{itemize}

  \printbibliography
\end{document}